\documentclass{amsart}[11pt]
\usepackage{amsmath, amsfonts, amsthm, amssymb,mathtools}
\usepackage{subfigure}
\usepackage{array}
\usepackage{diagbox}
\usepackage{stmaryrd}
\usepackage{verbatim}
\usepackage{hyperref}
\usepackage{color}
\usepackage{ulem}
\usepackage[dvips]{graphicx}
\usepackage{enumerate}
\usepackage{mathrsfs}
\numberwithin{equation}{section}
\newtheorem{theorem}{Theorem}[section]

\newtheorem{lemma}[theorem]{Lemma}

\theoremstyle{definition}
\newtheorem{definition}[theorem]{Definition}
\newtheorem{remark}[theorem]{Remark}
\newtheorem{algorithm}[theorem]{Algorithm}
\newtheorem{assumption}[theorem]{Assumption}

\newcommand{\ind}{1\hspace{-2.1mm}{1}} 

\newcommand{\eps}{\varepsilon}
\newcommand{\RR}{\mathbb{R}}
\newcommand{\EE}{\mathbb{E}}
\newcommand{\BS}{\mathrm{BS}}
\newcommand{\CC}{\mathrm{C}}
\newcommand{\TT}{\mathrm{T}}
\newcommand{\St}{\widetilde{S}}
\newcommand{\yy}{\mathrm{y}}
\newcommand{\Tf}{\mathfrak{T}}
\newcommand{\Kf}{\mathfrak{K}}
\newcommand{\Km}{\mathrm{K}}
\newcommand{\Pf}{\mathfrak{P}}
\newcommand{\ww}{\mathrm{w}}
\newcommand{\DD}{\mathrm{D}}

\newcommand{\xx}{\mathrm{x}}
\newcommand{\pp}{\mathrm{p}}
\newcommand{\qq}{\mathrm{q}}

\newcommand{\D}{\mathrm{d}}

\newcommand{\E}{\mathrm{e}}

\newcommand{\Nn}{\mathcal{N}}
\newcommand{\Pp}{\mathcal{P}}
\newcommand{\Xx}{\mathcal{X}}

\newcommand{\AAA}{\mathrm{A}}
\newcommand{\NN}{\mathbb{N}}

\newcommand{\argmin}{\mathrm{arg}\,\mathrm{min}}
\newcommand{\argmax}{\mathrm{arg}\,\mathrm{max}}

\newcommand{\Oo}{\mathcal{O}}

\begin{document}

\title{No-arbitrage bounds for the forward smile given marginals}
\author{Sergey Badikov, Antoine Jacquier, Daphne Qing Liu and Patrick Roome}
\address{Department of Mathematics, Imperial College London}
\email{sergey.badikov08@imperial.ac.uk, a.jacquier@imperial.ac.uk, q.liu12@imperial.ac.uk, p.roome11@imperial.ac.uk}
\date{\today}
\keywords {martingale optimal transport, robust bounds, forward-start, heston}
\subjclass[2010]{91G20, 91G80, 90C46, 90C05, 90C34}
\thanks{
The authors are indebted to Claude Martini for stimulating discussions,
and to the anonymous referees for helpful suggestions.
AJ acknowledges financial support from the EPSRC First Grant EP/M008436/1.
}

\maketitle
\begin{abstract}
We explore the robust replication of forward-start straddles given quoted (Call and Put options) market data. 
One approach to this problem classically follows semi-infinite linear programming arguments, 
and we propose a discretisation scheme to reduce its dimensionality and hence its complexity.
Alternatively, one can consider the dual problem, consisting in finding optimal martingale measures
under which the upper and the lower bounds are attained. 
Semi-analytical solutions to this dual problem were proposed by Hobson and Klimmek~\cite{HK13} 
and by Hobson and Neuberger~\cite{HN08}.
We recast this dual approach as a finite dimensional linear programme, 
and reconcile numerically, in the Black-Scholes and in the Heston model, the two approaches.
\end{abstract}

\section{Introduction}

Since David Hobson's seminal contribution~\cite{H98}, an important stream of literature has focused on developing model-free sub(super)-hedges for multi-dimensional derivative products or path-dependent options, given a set of European option instruments.
The key observation is that the model-free sub(super)-hedging cost is closely related to the Skorokhod Embedding problem 
(see the exhaustive survey papers by Hobson~\cite{H11} and Ob\l\'oj~\cite{O4} 
in the context of mathematical finance).
Recently, Beiglb\"ock, Henry-Labord\`ere and Penkner~\cite{BHP11} studied this problem in the framework of martingale optimal transport theory.
Assuming that European Call/Put option prices are known for all strikes and some maturities 
(equivalently the marginal distributions of the asset price are known at these times), 
optimal transport then yields a no-arbitrage range of prices of a derivative product consistent with these marginal distributions.
The primal problem endeavours to find the supremum, or the infimum, 
of these prices over all joint martingale measures (transport plans) consistent with the marginals.
The dual problem, in turn, seeks to find the `best' sub(super)-replicating portfolio;
this dual formulation has the advantage of a natural financial interpretation and can be cast as an (infinite) 
linear programme, amenable to numerical implementations, as proposed in~\cite{PHL11}.

Forward-start options (of Type I and of Type II) are among the simplest products amenable to these techniques.
The upper bound price for the at-the-money Type-II forward-start straddle 
(with payoff $|S_{t+\tau} - S_{t}|$ for some $t, \tau>0$)
was computed by Hobson and Neuberger~\cite{HN08},
where the support of the optimal martingale measure is a binomial tree.
Unfortunately the optimal measure and the associated super-hedging portfolio are not available analytically.
The martingale optimal transference plan for the lower bound price of the at-the-money Type-II forward-start straddle has been characterised semi-analytically by Hobson and Klimmek~\cite{HK13}, 
and the transference plan (supported along a trinomial tree) is found by solving a set of coupled ODEs.
Recently, Campi, Laachir and Martini~\cite{CLM14} studied the change of numeraire in these two-dimensional optimal transport problems and showed 
that, under some technical conditions, the lower bound for the Type-I at-the-money forward-start straddle is also attained by the Hobson-Klimmek transference plan.

In this paper, we numerically investigate the no-arbitrage bounds of the Type-II forward-start straddle.
The infinite-dimensional linear programme corresponding to the optimal transport problem
is presented in Section~\ref{sec:probformopttransp}.
Section~\ref{sec:noarbdiscretopttrans} focuses on a reduction of the dimension of the problem,
by discretising the support of the marginal distributions at times~$t$ and~$t+\tau$, which respects the consistency of the primal and the dual problems with observed option prices, and yields robust numerical results.
Our discretisation method differs from that of Henry-Labord\`ere~\cite{PHL11},
and requires far fewer points, thus reducing the complexity of the finite-dimensional linear programmes to be solved,
thereby improving the algorithmic speed.
In Section~\ref{sec:OTnumerics}, we specialise our computation of the upper and lower bounds to the cases where the marginal distributions are generated from 
a Black-Scholes model (lognormal marginals) and from the Heston stochastic volatility model. 
In the lower bound at-the-money case we numerically solve, in Section~\ref{sec:primalnumerics}, 
the coupled ordinary differential equations associated with the Hobson-Klimmek transference plan, 
and show that it is in striking agreement with the LP solution of the dual problem.
In Section~\ref{sec:numericTPs} we numerically solve the primal problem 
and provide the optimal transport plans for a large range of strikes.
The transport plans are only known 
for the at-the-money case~\cite{HK13,HN08}, and we highlight numerical evidence that the optimal transference plans are more subtle in these cases, 
and appear to be a combination of the lower and upper bound at-the-money plans.
Intuitively, the extremal measure yields a price corresponding either to the maximum or to the minimum value of the product.
In the forward-start option case, this extremal measure maximises or minimises the kurtosis of the conditional distribution of the asset price process (see in particular Sections~\ref{sec:primalnumerics} and~\ref{sec:numericTPs} and~\cite{HK13,HN08}).
Therefore a choice of a model that misspecifies the kurtosis might lead to wrong risk exposure profile. 


In the examples explored in Section~\ref{sec:OTnumerics} the range of forward smiles consistent 
with the marginal laws is large, and of the same magnitude as in~\cite[Section 5.5]{PHL11}.
This wide range of prices supports the claim that using European vanilla options to replicate forward volatility-dependent claims seems illusory.
Forward-start options should therefore be seen as fundamental building blocks for exotic option pricing and not decomposable (or approximately decomposable) into European options.
Forward-start options are usually available from the stripping of cliquet options, on OTC markets.
Cliquet options have been a popular instrument on Equity markets~\cite{Buehler},
and also as a hedging tool for insurance companies (the so-called FLEX Index options, 
the details of which can be found on the CBOE website).
Models used for forward volatility-dependent exotics should be able to calibrate forward-start option prices 
and should produce realistic forward smiles consistent with trader expectations and observable prices.

\textbf{Notations:}
$B_b(\RR)$ denotes the set of bounded measurable functions on the real line,
$\RR_+:=[0,\infty)$ represents the positive half-line;
$\left\langle\cdot,\cdot \right\rangle$ denotes the Euclidean inner product,
and $\|\cdot\|$ the $L^1$ norm.

\section{Problem formulation}\label{sec:probformopttransp}

We consider an asset price process $(S_t)_{t\geq 0}$ starting at $S_0=1$,
and we assume that interest rates and dividends are null. 
For $t,\tau>0$,
let $\mu$ and $\nu$ denote the distributions of $S_t$ and $S_{t+\tau}$,
assumed to have common finite mean, supported on $[0,\infty)$, 
and absolutely continuous with respect to the Lebesgue measure.
We say that the bivariate law $\zeta$ is a martingale coupling, and write $\zeta\in\mathcal{M}(\mu,\nu)$, 
if~$\zeta$ has marginals~$\mu$ and~$\nu$ and $\int_{y\in\RR_+}(y-x)\zeta(\D x, \D y)=0$ for each $x\in\RR_{+}$.
Following~\cite{S65}, we shall say that~$\mu$ and~$\nu$ are in convex, or balayage, order
(which we denote $\mu\preceq\nu$) if they have equal means and satisfy 
$\int_{\RR_+}(y-x)^{+}\mu(\D y) \leq \int_{\RR_+}(y-x)^{+}\nu(\D y)$ for all~$x\in\RR$.
This assumption ensures (see~\cite{B12}) that the set $\mathcal{M}(\mu,\nu)$ is not empty.
Our objective is to find the tightest possible lower and upper bounds, 
consistent with the two marginal distributions, for the forward-start straddle payoff $|S_{t+\tau}-\Kf S_{t}|$ with $\Kf>0$.
To this end we define our primal problems:
\begin{equation}\label{eq:primal}
\underline{\Pp}(\mu,\nu) := \inf_{\zeta\in\mathcal{M}(\mu,\nu)}\int_{\RR_+^2}|y-\Kf x|\zeta(\D x, \D y), 
\qquad\text{and}\qquad
\overline{\Pp}(\mu,\nu) := \sup_{\zeta\in\mathcal{M}(\mu,\nu)}\int_{\RR_+^2}|y-\Kf x|\zeta(\D x, \D y).
\end{equation}
We now define the following sub- and super-replicating portfolios:
\begin{align*}
\underline{Q} & := \displaystyle \left\{(\psi_{0},\psi_1,\delta) \in L^1(\mu)\times L^1(\nu)\times B_b(\RR_+):
\psi_{1}(y)+\psi_{0}(x)+\delta(x)(y-x)\leq|y-\Kf x|, \text{for all }x,y \in\RR_{+}\right\},\\
\overline{Q} & := \displaystyle \left\{(\psi_{0},\psi_1,\delta) \in L^1(\mu)\times L^1(\nu)\times B_b(\RR_+):
\psi_{1}(y)+\psi_{0}(x)+\delta(x)(y-x)\geq|y-\Kf x|, \text{for all }x,y \in\RR_{+}\right\}.
\end{align*}
Clearly if $(\psi_0,\psi_1,\delta) \in \underline{Q}$ (respectively $\in \overline{Q}$) then
$\int_{\RR_+^2}|y-\Kf x|\zeta(\D x, \D y) \geq (\leq) 
\int_{\RR_+}\psi_0(x) \mu(\D x) + \int_{\RR_+}\psi_1(y) \nu(\D y)$.
The dual problems are then defined as the supremum (infimum) over all sub(super)-replicating portfolios:
\begin{equation}\label{eq:dual}
\begin{array}{ll}
\displaystyle \int_{\RR_+^2}|y-\Kf x|\zeta(\D x, \D y) \geq
\sup_{(\psi_0,\psi_1,\delta) \in \underline{Q}}
\left\{ \int_{\RR_+}\psi_0(x) \mu(\D x) + \int_{\RR_+}\psi_1(y) \nu(\D y) \right\}
 =: \underline{\mathcal{D}}(\mu,\nu), & \\ 
\displaystyle \int_{\RR_+^2}|y-\Kf x|\zeta(\D x, \D y) \leq
\inf_{(\psi_0,\psi_1,\delta) \in \overline{Q}}
 \left\{\int_{\RR_+}\psi_0(x) \mu(\D x) + \int_{\RR_+}\psi_1(y) \nu(\D y) \right\}
 =: \overline{\mathcal{D}}(\mu,\nu). &
\end{array}
\end{equation}
In~\cite[Theorem 1 and Corollary 1.1]{BHP11}, the authors proved 
(actually for a more general class of payoff functions) that there is no duality gap, namely that
both equalities $\underline{\Pp}(\mu,\nu) = \underline{\mathcal{D}}(\mu,\nu)$
and 
$\overline{\Pp}(\mu,\nu) = \overline{\mathcal{D}}(\mu,\nu)$ hold.
However, the optimal values may not be attained in the dual problems,
as proved in~\cite[Proposition 4.1]{BHP11}.
In~\cite{HK13} and~\cite{HN08} the authors showed that in the at-the-money case ($\Kf =1$),
with an additional dispersion assumption on the measures~$\mu$ and~$\nu$ 
the optimal values of the dual problems~\eqref{eq:dual} are actually attained.
More precisely, the following result summarises~\cite{HK13} 
(see also Section~\ref{sec:primalnumerics} below for more details about the technical assumption):

\begin{assumption}\label{assump:munuassump}
The support of $\eta:=(\mu-\nu)_+$ is given by an interval $[a,b]\subset\RR_+$,
and the support of $\gamma:=(\nu-\mu)_+$ is given by $\RR_{+}\setminus [a,b]$.
\end{assumption}

\begin{theorem}
The set equalities $\underline{\mathcal{D}}(\mu,\nu)=\underline{\Pp}(\mu,\nu)$
and $\overline{\mathcal{D}}(\mu,\nu)=\overline{\Pp}(\mu,\nu)$ hold,
and the primal optima in~\eqref{eq:primal} are attained:
there exist martingale measures $\mathbb{Q}_{\textrm{L}}$ and $\mathbb{Q}_{\textrm{U}}$
in~$\mathcal{M}(\mu,\nu)$
such that $\underline{\Pp}(\mu,\nu)=\mathbb{E}^{\mathbb{Q}_{\textrm{L}}}|S_{t+\tau}-\Kf S_{t}|$ 
and $\overline{\Pp}(\mu,\nu)=\mathbb{E}^{\mathbb{Q}_{\textrm{U}}}|S_{t+\tau}-\Kf S_{t}|$.
Furthermore, under Assumption~\ref{assump:munuassump}, the infimum and supremum 
in the dual problems~\eqref{eq:dual} are attained when $\Kf = 1$.
\end{theorem}


\section{No-Arbitrage discretisation of the primal and dual problems}\label{sec:noarbdiscretopttrans}

\subsection{No-arbitrage discretisation of the density}\label{sec:DiscrCons}
Let $t>0$ be some given time horizon, $S_t$ the random variable describing the stock price
at time~$t$, and $\mu$ the law of~$S_t$.
Fix $m>1$ and suppose that we are given a set~$\xx = (x_1,\ldots,x_m)\in\RR^{m}_+$ of points $0<x_1<x_2<...<x_m$ in the support of~$\mu$, and a 
discrete distribution~$\qq$ 
with atom $q_i$ at the point~$x_i$ sampled from $\mu$ 
(for example if $\mu$ admits density $f$ then one can take $q_i=f(x_i)/\sum_{i=1}^mf(x_i)$).
We wish to find a discrete distribution~$\pp$, close to~$\qq$, 
matching the first $l\leq m$ moments of~$\mu$, 
in particular satisfying the `martingale' condition $\langle \pp, \xx\rangle = 1$.
Let $\TT:\RR_{+}\to\RR_{+}^{l}$ be given
by $\TT(x):=(x,x^2,...,x^{l})$
and define the moment vector $\Tf := \int_{\RR_{+}}\TT(x)\mu(\D x) \in \RR_+^{l}$.
Such a matching condition is not necessarily consistent with a given set of (European) option prices.
In order to ensure that the discrete density re-prices these options, 
we add a second layer:
Borwein, Choksi and Mar\'echal~\cite{Borwein} suggested to recover discrete probability distributions from observed market prices of European Call options by
minimising the Kullback-Leibler divergence to the uniform distribution 
(they also comment that any prior distribution can be chosen). 
In particular given the law~$\mu$ of~$S_t$ and a set of European Call option prices
$\Pf \in \RR^M_+$,
maturing at~$t$ with strikes $K_1<\cdots<K_M$, we can solve the minimisation problem:
\begin{equation}\label{eq:moomp}
\min_{\pp\in [0, 1]^{m}: \|\pp\| = 1}\sum_{i=1}^{m}p_i\log\left(\frac{p_i}{q_i}\right),
\qquad
\text{subject to }
\left(\left(\langle\CC_j(\xx),\pp\rangle\right)_{j=1,\ldots,M}, 
\left(\langle\TT_j(\xx),\pp\rangle\right)_{j=1,\ldots,l}\right) = (\Pf,\Tf).
\end{equation}
for some prior discrete distribution~$\qq$,
where $\CC_j(\xx):=\left((x_1-K_j)_+,\ldots,(x_m-K_j)_+\right)$, for $j=1,\ldots,M$, denotes the payoff vector of the options
and $\TT_j(\xx) := (x_1^j,\ldots,x_m^j)$ is the $j$-th moment vector ($j=1,\ldots,l$).
\begin{definition}
The discrete distribution~$\pp$ is consistent with~$\Pf$
whenever the following hold:
\begin{enumerate}[(i)]
\item $\sum_{i=1}^{m}p_{i}=1$;
\item $\langle \CC_j(\xx),\pp\rangle = \sum_{i=1}^{m}p_{i}\CC_j(x_i) = \Pf_j$ for all $j=1,\ldots,M$;
\item $\langle \TT_1(\xx),\pp\rangle = \langle\xx, \pp\rangle = \Tf_1 = 1$;
\end{enumerate}
\end{definition}
Note that the last item in the definition above is nothing else than the martingale condition.
It must be noted that if the full marginal distribution~$\mu$ of~$S_t$ is known, 
any finite subset of European options can be chosen above 
and the price vector can be defined as $\Pf := \int_{\RR_+}\CC(x)\mu(\D x)$,
where $\CC(x) := ((x-K_1)_+, \ldots, (x-K_M)_+)\in \RR_+^M$ for any $x\in\RR_+$.
In particular the solution to this problem can be obtained as a modification of the solution in~\cite{TT13},
which itself is based on arguments by Borwein and Lewis~\cite[Corollary 2.6]{BorweinLewis}:
\begin{equation}\label{eq:moomsol}
p_i = \displaystyle \frac{q_i\exp\Big(\left\langle \lambda^*, \left(\CC(x_i),\TT(x_i)\right)\right\rangle\Big)}{\sum_{i=1}^{m}q_i \exp\Big(\left\langle \lambda^*, \left(\CC(x_i), \TT(x_i)\right)\right\rangle\Big)},
\quad
\text{where}
\quad
\lambda^*:=\underset{\lambda\in\RR^{M+l}}{\mathrm{argmin}}
\left\{-\left\langle \lambda, \left(\Pf,\Tf\right)\right\rangle
 + \log\left(\sum_{i=1}^{m}q_i \E^{\left\langle \lambda, \left(\CC(x_i), \TT(x_i)\right)\right\rangle}\right)\right\}.
\end{equation}
We can now consider the following algorithm:

\begin{algorithm}\label{algo:generatediscrete}\ 
\begin{enumerate}[(i)]
\item 
Several choices are possible for the~$m$ points $0<x_1<...<x_m$; for instance:
\begin{enumerate}[(a)]
\item \textit{Binomial}: 
Let~$\Sigma$ denote the at-the-money lognormal volatility (for European options maturing at~$t$).
Set $\delta := t/(m-1)$, $u:=1 + \left(\E^{\delta\Sigma^2}-1\right)^{1/2}$,
$d:=1 - \left(\E^{\delta\Sigma^2}-1\right)^{1/2}$,
and $x_{i}:=u^{i-1} d^{m-i}$;
\item \textit{Gauss-Hermite}: 
$x_i:=\E^{\tilde{x}_i}$, where 
$\tilde{x}_1,...,\tilde{x}_{m}$ are the nodes of an $N$-point Gauss-Hermite quadrature.
\end{enumerate}
\item For the discrete distribution~$q$, we can follow several routes:
\begin{enumerate}[(a)]
\item if~$\mu$ admits a density~$f_\mu$, then, for $i=1,\ldots,m$, 
set $q_i := f_\mu(x_i)/\sum_{j=1}^{m}f_\mu(x_j)$;
\item alternatively, for $i=1,...,m$, let $q_i:=\mu([x_{i-1},x_{i}))$ (with $x_0=0$);
\end{enumerate}
\item Compute the discretised measure~$p$ through~\eqref{eq:moomsol}.
\end{enumerate}
\end{algorithm}

\begin{remark}
As pointed out by Tanaka and Toda~\cite{TT13} the choice of discretisation points $x_1,\ldots,x_m$ in Algorithm~\ref{algo:generatediscrete} is dictated by the discretisation of the integrals
$\int_{\RR_{+}}(\CC(x),\TT(x))\D f_{\mu}(x) \approx \sum_{i=1}^m w(x_i)(\CC(x_i),\TT(x_i))f_{\mu}(x_i)$. 
The weights~$w(\cdot)$ are chosen in accordance with a given quadrature rule; 
in the case of Algorithm~\ref{algo:generatediscrete} the weights are chosen constant 
$w(x_i) = \left(\sum_{i=1}^m f_{\mu}(x_i)\right)^{-1}$ 
for all $i = 1,\ldots,m$.
\end{remark}

In the particular case where the discretisation nodes and the given strikes satisfy some specific ordering, 
it is possible to choose a more explicit weighting scheme~$\pp$ consistent with~$\Pf$
with minimal assumptions on the choice of discretisation:

\begin{lemma}\label{lemma:M2discretisation}
Let $m = M+2$.
Consistency with $\Pf$ is ensured if both sets of conditions hold:
\begin{enumerate}[(i)]
\item $x_1 < K_1$, $x_{M+2} \geq \left(\Pf_{M-1}K_M - \Pf_{M}K_{M-1}\right) / \left(\Pf_{M-1}-\Pf_M\right)$ and
$x_{i+1} = K_i$ for $i=1,\ldots,M$;
\item $p_{M+2} = \Pf_M/\left(x_{M+2} - K_M\right)$, $p_{1} = 1 - \displaystyle\sum_{j=2}^{M+2}p_j$, 
and, for any $i = M,\ldots,1$,
$$
p_{i+1} = \frac{1}{x_{i+1} - K_{i-1}}\left(\Pf_{i-1} - \sum_{j=i+2}^{M+2} p_{j}(x_{j} - K_{i-1})\right),
$$
\end{enumerate}
with the convention $K_0:=x_1$ and $\Pf_0 := 1-x_1$.
\end{lemma}

\begin{proof}
Since the strikes are ordered and the vector
$\overline{\Pf}:=(1, 1, \Pf_1, \ldots, \Pf_M)' \in \RR^{M+2}$ satisfies no-arbitrage conditions
(in the sense of~\cite[Theorem 3.1]{DH07}), the assumption
$x_{M+2} \geq \left(\Pf_{M-1}K_M - \Pf_{M}K_{M-1}\right) / \left(\Pf_{M-1}-\Pf_M\right)$
implies that $x_{M+2} > K_M$.
The consistency of the weighting scheme with~$\Pf$ can be written as $\AAA\pp^\xx = \overline{\Pf}$, with
\begin{equation*}
\AAA := \begin{pmatrix}
1 & 1 & \ldots & \ldots & 1\\
x_1 & x_2 & \ldots & \ldots & x_{M+2}\\
0 & (x_2 - K_1) & \ddots & \ddots & (x_{M+2} - K_1)\\
\vdots & \ddots & \ddots & \ddots & \vdots\\
0 & \ddots & 0 & (x_{M} - K_{M-1}) & (x_{M+2} - K_{M-1})\\
0 & \ldots & 0 & 0 & (x_{M+2} - K_M)
\end{pmatrix}\qquad\text{and}\qquad
\pp^\xx := 
\begin{pmatrix}
p_1 \\
p_2 \\
\vdots \\
p_{M} \\
p_{M+2}
\end{pmatrix}.
\end{equation*}
Here~$\AAA$ is a real upper triangular matrix, 
so that the system has a unique solution given in the lemma. 
It remains to check that the feasible solutions of the system satisfy the additional constraints
$p_i \geq 0$ for $i=1,\ldots,M+2$.

Since $x_{M+2} \geq \frac{\Pf_{M-1}K_M - \Pf_{M}K_{M-1}}{\Pf_{M-1}-\Pf_M}$ it follows that 
$$
p_{M+2} = \frac{\Pf_M}{x_{M+2} - K_M}\leq \frac{\Pf_{M-1}-\Pf_M}{K_M-K_{M-1}}
\leq \frac{-K_{M-1} + K_M}{K_M-K_{M-1}},
$$
where the last inequality follows from the Put-Call Parity and absence of arbitrage in~$\overline{\Pf}$.
By definition we have 
$$
p_{M+1} = \frac{1}{K_M - K_{M-1}}\left(\Pf_{M-1} - p_{M+2}(x_{M+2} - K_{M-1})\right) =
\frac{1}{K_M - K_{M-1}}\left(\Pf_{M-1} - \Pf_{M} - \Pf_{M}\frac{K_M - K_{M-1}}{x_{M+2} - K_{M}}\right),
$$
and by assumption on $x_{M+2}$ we have that $p_{M+1}\geq 0$. 
Similarly for $p_{M}$ we have 
\begin{align*}
p_{M} & = \frac{1}{K_{M-1} - K_{M-2}}\left[\Pf_{M-2} - p_{M+1}(K_{M} - K_{M-2}) - p_M(x_{M+2} - K_{M-2})\right]\\
 & = \frac{\Pf_{M-2}-\Pf_{M-1}}{K_{M-1} - K_{M-2}} - \frac{\Pf_{M-1}-p_{M+2}(x_{M+2} - K_M)}{K_{M} - K_{M-1}}\\
 & = \frac{\Pf_{M-2}-\Pf_{M-1}}{K_{M-1} - K_{M-2}} - \frac{\Pf_{M-1}-\Pf_{M}}{K_{M} - K_{M-1}}.
\end{align*}
Proceeding recursively we observe that for $i=2,\ldots,M-1$ we have 
$$
p_{i} = \frac{\Pf_{i-2}-\Pf_{i-1}}{K_{i-1} - K_{i-2}} - \frac{\Pf_{i-1}-\Pf_{i}}{K_{i} - K_{i-1}}. 
$$
Since the prices $\overline{\Pf}$ satisfy no-arbitrage conditions, 
then $p_i \geq 0$ for $i=1,\ldots,M+2$.

Note that, when $x_{M+2} = \frac{\Pf_{M-1}K_M - \Pf_{M}K_{M-1}}{\Pf_{M-1}-\Pf_M}$, 
$$
p_{M+2} = \frac{\Pf_{M-1}-\Pf_M}{K_M-K_{M-1}},
\qquad
p_{M+1} = 0, 
\qquad 
p_{i} = \frac{\Pf_{i-2}-\Pf_{i-1}}{K_{i-1} - K_{i-2}}-\frac{\Pf_{i-1}-\Pf_{i}}{K_{i} - K_{i-1}},
$$
for $i=2,\ldots,M$ and the discretisation reduces to $M+1$ points, where $x_{M+1} = K_M$ is discarded.
\end{proof}

\begin{remark}
One could in principle generalise Lemma~\ref{lemma:M2discretisation} to a construction where $x_{i+1}\in(K_{i-1},K_i]$ for $i=1,\ldots,M$ (and obviously $x_1<K_1$ and $x_{M+2}>K_M$). 
Quick computations however reveal that writing a clear set of sufficient conditions for consistency 
with~$\overline{\Pf}$ is notationally cumbersome and practically not particularly enlightening.
\end{remark}

\subsection{Balayage-consistent discretisation}

Assume now that there are~$M$ European Call options~$\Pf^\xx$ available with maturity~$t$ 
and~$N$ options~$\Pf^\yy$ available with maturity~$t+\tau$, 
and that the measures~$\mu$ and~$\nu$ are calibrated to those European options:
$\int_{\RR_+}(x-K^\xx_i)_+\mu(\D x) = \Pf^\xx_i$ for all $i=1,\ldots,M$ 
and $\int_{\RR_+}(y-K^\yy_j)_+\nu(\D y) = \Pf^\yy_j$ for all $j=1,\ldots,N$. 
We further assume that~$\Pf^\xx$ and~$\Pf^\yy$ do not contain arbitrage 
(in the sense of~\cite[Theorem 4.2]{DH07}), 
otherwise there is no equivalent martingale measure and the primal problem~\eqref{eq:primal} is infeasible.
For two discretisation meshes $\xx = (x_1,\ldots,x_m)\in\RR^m_+$ and $\yy = (y_1,\ldots,y_n)\in\RR^n_+$,
Algorithm~\ref{algo:generatediscrete}, for instance, produces discrete distributions~$\pp^\xx$ 
and~$\pp^\yy$ supported on~$\xx$ and~$\yy$.
It is not however guaranteed that the convex ordering of the original measures~$\mu$ and~$\nu$ 
is preserved for~$\pp^\xx$ and~$\pp^\yy$.
The following lemma provides sufficient conditions ensuring this.

\begin{lemma}\label{lemma:weightsconst}
Let~$\pp^\xx$ and~$\pp^\yy$ be two discrete measures supported on~$\xx$ and~$\yy$.
The following conditions altogether ensure that $\pp^\xx\preceq\pp^\yy$:
\begin{enumerate}[(i)]
\item $\pp^\xx$ and~$\pp^\yy$ have the same mean equal, $y_1\leq x_1$ and $y_n\geq x_m$;
\item Assumption~\ref{assump:munuassump} holds with~$\mu = \pp^\xx$ and~$\nu = \pp^\yy$.
\end{enumerate}
\end{lemma}
\begin{remark}
In our framework, the discrete measures~$\pp^\xx$ and~$\pp^\yy$ 
arise as discretisations of the original measures~$\mu$ and~$\nu$.
If~$\pp^\xx$ and~$\pp^\yy$ are to be consistent with~$\Pf^\xx$ and~$\Pf^\yy$,
then the discretisation nodes~$\xx$ and~$\yy$ must be finer than the sets of input strikes,
which we can write as
\begin{itemize}
\item $x_1<K^\xx_1$, $y_1<K^\yy_1$, $x_m > K^\xx_M$ and $y_n > K^\yy_N$;
\item for all $i\in\{1,\ldots,M-1\}$, there exists $k_i \in \{2,\ldots,m-1\}$ such that $K^\xx_i\leq x_{k_i}\leq K^\xx_{i+1}$,
and for all $j\in\{1,\ldots,N-1\}$, there exists $k_j\in \{2,\ldots,n-1\}$ such that $K^\yy_j\leq y_{k_j}\leq K^\yy_{j+1}$.
\end{itemize}
\end{remark}

\begin{proof}[Proof of Lemma~\ref{lemma:weightsconst}]
To fix the notations,
for any set~$A\subset\RR$, define $\pp^\xx(A) := \sum_{\{i:x_i\in A\}}p_i^\xx$ and 
$\pp^\yy(A) := \sum_{\{j: y_j\in A\}}p_j^\yy$.
In particular for any $z\in[0,y_n]$ we have 
$$
\pp^\yy([0,z]) = \sum_{\{j\leq n: y_j\leq z\}}p^\yy_j,
\qquad\text{and}\qquad
\pp^\xx([0,z]) = \sum_{\{i\leq m: x_i\leq z\}}p^\xx_i.
$$
Define furthermore the function $\delta F: \RR_+\to\RR$ by $\delta F(z) := \pp^\yy([0,z]) - \pp^\xx([0,z])$.
Assumption~\ref{assump:munuassump} with~$\mu = \pp^\xx$ and~$\nu = \pp^\yy$,
together with the boundary conditions in Lemma~\ref{lemma:weightsconst}(i),
can therefore be written as 
\begin{equation}\label{eq:etagamma}
\left\{
\begin{array}{ll}
\displaystyle \eta(A) = (\pp^\xx(A) - \pp^\yy(A))_+>0, & \text{for any }A\subseteq [a,b],\\
\displaystyle \gamma(A) = (\pp^\yy(A) - \pp^\xx(A))_+ >0, & \text{for any }A\subseteq [0,y_n]\setminus[a,b].
\end{array}
\right.
\end{equation}

Define now the functions $f_{\xx},g_{\yy}:\RR_+\to\RR_+$ by
$f_{\xx}(K) := \langle\pp^{\xx}, (\xx - K)_+\rangle$ 
and $g_{\yy}(K) := \langle\pp^{\yy}, (\yy - K)_+\rangle$.
By~\cite[Chapter~2, Definition~2.1.6]{B12}, 
the balayage condition $\pp^\xx\preceq\pp^\yy$ will be satisfied as soon as 
$f_{\xx}$ and $g_{\yy}$ are convex and $f_{\xx}(K)\leq g_{\yy}(K)$ for all $K\in\RR_+$. 
Since~$\pp^\xx$ and~$\pp^\yy$ have non-negative components,
the functions~$f_\xx$ and~$g_\yy$ are clearly convex, 
so we are left to prove that $f_{\xx}(\cdot)\leq g_{\yy}(\cdot)$ on~$\RR_+$.

We first show that the conditions 
$y_1 \leq x_1$ and $x_m \leq y_n$ are necessary.
If $x_1<y_1$, then $f_{\xx}(K) = g_{\yy}(K) = 1 - K$ for all $K\in [0,x_1]$.
Also, for any $K\in (x_1, x_2\wedge y_1)$, 
$g_{\yy}(K) = 1 - K$ and $f_{\xx}(K) = \sum_{i=2}^mp^\xx_i(x_i - K) = (1-K) + p^\xx_1(K-x_1)>(1-K)$, 
which violates the balayage order. 
Similarly, if $x_m > y_n$, let $i^* :=\inf\{i : x_i\geq y_n\}$;
then for any $K\in (y_n,x_m)$,
$g_{\yy}(K) = 0$ and $f_{\xx}(K) = \sum_{i=i^*}^mp^{\xx}_i(x_i - K)$
which yields the conclusion. 

Introduce now the function $G : [0,y_n] \rightarrow \RR$ as
\begin{align}\label{eq:FunctionG}
G(K) :=g_\yy(K) - f_\xx(K)
 & = \sum_{\{j: y_j>K\}}p^\yy_j(y_j - K) - \sum_{\{i: x_i>K\}}p^\xx_{i}(x_i - K)\nonumber\\
 & = \left(\sum_{\{i: x_i>K\}}p^\xx_{i} - \sum_{\{j: y_j>K\}}p^\yy_j\right)K
   + \sum_{\{j: y_j>K\}}p^\yy_jy_j - \sum_{\{i: x_i>K\}}p^{\xx}_{i} x_i.  
\end{align}
The function~$G$ is piecewise linear on $[0, y_n]$, 
not differentiable at the points $\{x_i\}_{1\leq i\leq m}$ and $\{y_j\}_{1\leq j\leq n}$,
attains its maximum and minimum on $[0,y_n]$ and $G(0) = G(y_n) = 0$. 
The lemma then follows if $G(K)\geq 0$ for all $K\in [0,y_n]$.
Let~$\partial^{+}G$ denote its right derivative on~$[0,y_n]$ 
(with the convention $\partial^{+}G(y_n)=0$).
The positivity of~$G$ together with the boundary conditions at~$0$ and~$y_n$ 
are therefore equivalent to the following:
$$
\partial^{+}G(K)\geq 0, \quad \text{for all }K\leq K^*
\qquad\text{and}\qquad
\partial^{+}G(K)\leq 0, \quad \text{for all }K> K^*,
$$
where $K^*:=\argmax_KG(K)$.
From~\eqref{eq:FunctionG}, we can therefore write, for any $K \in [0,y_n]$,
$$
\partial^{+}G(K) = \sum_{\{i: x_i>K\}}p^\xx_i - \sum_{\{j: y_j>K\}}p^\yy_j = 
\sum_{\{j: y_j \leq K\}}p^\yy_j - \sum_{\{i: x_i \leq K\}}p^\xx_i = \delta F(K),
$$
with~$\delta F$ defined above.
Since~$\delta F$ is piecewise constant, 
it remains to show that it admits a maximum and a minimum attained on unique subsets of $[0,y_n]$.
Note first that $\delta F(0) = 0 = \delta F(y_n)$, 
and that the image of~$[0,y_n]$ by~$\delta F$ is exactly~$[-1,1]$.
For any $z\in [0,a)\cup (b,y_n]$, $\delta F(z) = \gamma([0,z])$, and therefore
$\delta F$ is an increasing function on $[0,a)\cup (b,y_n]$. 
On~$[a,b]$, however, one has 
$$
\delta F(z) = \pp^\yy([0,z]) - \pp^\xx([0,z]) = 
\pp^\yy([a,z]) - \pp^\xx([a,z]) + \pp^\yy([0,a)) - \pp^\xx([0,a)) =
-\eta([a,z]) + \gamma([0,a)).
$$
Since $\gamma([0,a))>0$ and $\eta$ is a measure on $[a,b]$, 
then $\delta F$ is decreasing on $[a,b]$.   
Therefore there exists a unique quadruplet~$(c, d, e, f)$ such that 
$[c,d]\subset [0,y_n]$, $[e,f]\subset [0,y_n]$ and
$[c,d] = \argmax_z\delta F(z)$ and $[e,f] = \argmin_z\delta F(z)$.
It therefore follows that $\delta F(\cdot)\geq 0$ on~$[0,z^*)$
and $\delta F(\cdot)\leq 0$ on~$[z^*,y_n]$, 
for some $z^*\in(d,e)$.
\end{proof}

\begin{remark}\label{remark:dispersionAssumptionComment}
Note that the dispersion Assumption~\ref{assump:munuassump} intuitively implies that the variance of the underlying price process increases with maturity,
which is the case for all stochastic volatility models.
Lemma~\ref{lemma:weightsconst} provides general conditions to ensure that convex order is preserved
under discretisation of continuous measures. 
These are however difficult to verify analytically for general distributions, 
or even for those obtained by Algorithm~\ref{algo:generatediscrete}.
In Figure~\ref{fig:convordercheck} below, we provide numerical evidence that the convex order is preserved for these distributions.
\end{remark}
\begin{figure}[htb] 
\centering
\includegraphics[scale=0.5]{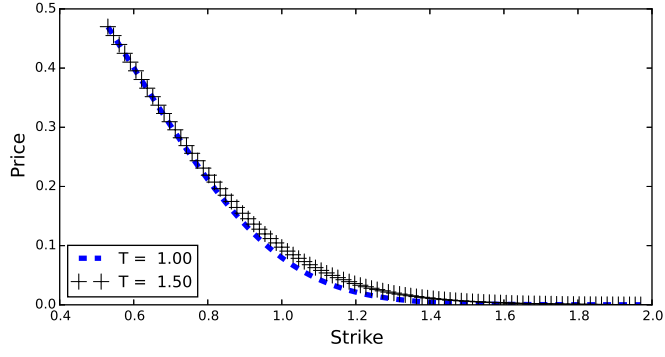}
\caption{
Call prices with maturities $t=1$ (dashed) and $t+\tau = 1.5$ (crosses) computed 
using the discretised densities of the Black-Scholes model, with parameter $\sigma = 0.2$.
As a consistency check, the Call prices with maturity $t+\tau$ are strictly greater than those with maturity~$t$,
and both functions are convex.
}
\label{fig:convordercheck}
\end{figure}

\subsection{Primal and dual formulation}\label{sec:DualDiscretisation}
We focus here on the discretisation of the primal and dual formulations for the upper bound, 
and note that an analogous formulation holds for the lower bound.
We use Algorithm~\ref{algo:generatediscrete} to approximate 
the random variables~$S_t$ and~$S_{t+\tau}$ by discrete random variables~$\St_t$ and~$\St_{t+\tau}$
with distributions~$\pp^\xx$ and~$\pp^\yy$ supported on~$\xx \in \RR^{m}_+$ and $\yy \in\RR^{n}_+$.
The linear programme for the primal problem~\eqref{eq:primal} then reads
\begin{align}\label{eq:DiscretisedPrimal}
 & \overline{\Pp}(\mu,\nu) := \max_{\zeta \in \mathscr{M}^{m,n}_+}\sum_{i,j}\zeta_{i,j}|y_j-\Kf x_i|,\\
 &
\text{subject to }
(\|\zeta_{i,\cdot}\|)_{i=1,\ldots,m} = \pp^\xx,\quad
(\|\zeta_{\cdot,j}\|)_{j=1,\ldots,n} = \pp^\yy,\quad
(\langle\zeta_{i,\cdot},(x_{i} - \yy)\rangle)_{i=1,\ldots,m} = \mathbf{0}\in\RR^m,
\nonumber
\end{align}
where $\mathscr{M}^{m,n}_+$ denotes the set of matrices of size $m\times n$ with non-negative entries.
For the dual problem, denote the Call option price (with strike~$K$ and maturity~$t$) 
on $\widetilde{S}_t$ by 
$\widetilde{C}(t,K) := \mathbb{E}(\widetilde{S}_t-K)_+ = \sum_{i=i^*}^{m}(x_i-K) p_i^\xx$, 
where $i^*:=\inf\{1\leq i\leq m: x_i>K\}$ and $\widetilde{C}(t,K)=0$ if $x_m\leq K$.
The next result is essential for the discretisation of the dual problem,
and follows by simple, yet careful, manipulations of telescopic sums.
\begin{lemma}
Let $l\in\NN$ and $\Km = (K_1,\ldots,K_{l})$.
For any one-dimensional real random variable~$Z$, the following representation holds almost surely:
\begin{equation}\label{eq:replicationdiscrete}
\varphi(Z) =
\varphi(K_1)+\DD\varphi(K_1)\left(Z - K_1\right)_+
+\sum_{i=2}^{l-1}\left(\DD\varphi(K_i)-\DD\varphi(K_{i-1})\right)
\left(Z - K_i\right)_+,
\end{equation}
for any continuous function~$\varphi$, where the forward finite-difference operator $\DD$ is defined as
$$
\DD\varphi(K_i):=\frac{\varphi(K_{i+1}) - \varphi(K_{i})}{K_{i+1} - K_i},
\qquad\text{for }i=1,\ldots,l-1.
$$ 
\end{lemma}

Let now $Z = \widetilde{S}_t$, $\varphi \equiv \psi_0$, and consider the vector of strikes $\Km^{\xx} \in \RR_{+}^{M}$.
Equality~\eqref{eq:replicationdiscrete} can then be rewritten as 
$
\psi_0(\widetilde{S}_t)=
\displaystyle w_0^{\xx}+w_1^{\xx} (\widetilde{S}_t - K_1^{\xx})_+
 + \sum_{i=2}^{M-1}w_{i}^{\xx}(\widetilde{S}_t - K_i^{\xx})_+
$,
with $l=M$ and where the weights~$w^{\xx}$ read
$$
w^{\xx}_0 := \psi_0(K_1^{\xx}), \qquad 
w^{\xx}_1:= \DD\psi_0(K_1^{\xx}),\qquad
w^{\xx}_i:= \DD\psi_0(K_i^{\xx})-\DD\psi_0(K_{i-1}^{\xx}),\text{ for } i=2,\ldots,M-1,
$$
so that
$$
\mathbb{E}\psi_0(\widetilde{S}_t)
 = w_0^{\xx} + w_1^{\xx} \EE\left(\widetilde{S}_t-K_1^{\xx}\right)_+  + \sum_{i=2}^{M-1}w_{i}^{\xx}\widetilde{C}(t,K_i^{\xx}).
$$
Likewise, for $Z = \widetilde{S}_{t+\tau}$, $\varphi \equiv \psi_1$ and $\Km^{\yy} \in \RR_{+}^{N}$,
an analogous formulation holds at time $t+\tau$:
$$
\mathbb{E}\psi_1(\widetilde{S}_{t+\tau})
 = w_0^{\yy} + w_1^{\yy} \EE\left(\widetilde{S}_{t+\tau}-K_1^{\yy}\right)_+  + \sum_{i=2}^{N-1}w_{i}^{\yy}\widetilde{C}(t+\tau,K_i^{\yy}),
$$
with
with the identification (from~\eqref{eq:replicationdiscrete}) $l=N$ and where the weights~$w^{\yy}$ read
$$
w^{\yy}_0 := \psi_0(K_1^{\yy}), \qquad 
w^{\yy}_1:= \DD\psi_0(K_1^{\yy}),\qquad
w^{\yy}_i:= \DD\psi_0(K_i^{\yy})-\DD\psi_0(K_{i-1}^{\yy}),\text{ for } i=2,\ldots,N-1.
$$

Define the set 
$\displaystyle 
\Xx := \{(x,y): x\in \mathrm{Supp}(\widetilde{S}_{t}), y\in\mathrm{Supp}(\widetilde{S}_{t+\tau})\}
$,
and assume from now on that 
$K_1^{\xx}$ and $K_1^{\yy}$ are such that
$\widetilde{S}_t \geq K_1^{\xx}$ and $\widetilde{S}_{t+\tau} \geq K_1^{\yy}$ almost surely
(equivalently, $K_1^{\xx}\leq x_1$ and $K_1^{\yy}\leq y_1$), 
so that the martingale property (ensured via Algorithm~\ref{algo:generatediscrete}) yields
$\EE(\widetilde{S}_t-K_1^{\xx})_+ = 1-K_1^{\xx}$
and 
$\EE(\widetilde{S}_{t+\tau}-K_1^{\yy})_+ = 1-K_1^{\yy}$.
The dual problem~\eqref{eq:dual} then reads
\begin{equation}\label{eq:SemiInfDual}
\overline{\mathcal{D}}(\mu,\nu)
 = \min
\left\{ v +w_1^{\xx}  + w_1^{\yy}
+ \sum_{i=2}^{M-1}w_i^{\xx} \widetilde{C}(t,K_i^{\xx})
+ \sum_{i=2}^{N-1}w_i^{\yy}\widetilde{C}(t+\tau,K_i^{\yy}): 
(\ww^{\xx},\ww^{\yy},\delta) \in \RR^{M+N}\times C_b(\RR_+)\right\},
\end{equation}
with 
$\ww^{\xx} = (w_0^{\xx},\ldots, w_{M-1}^{\xx})$,
$\ww^{\yy} = (w_0^{\yy},\ldots, w_{N-1}^{\yy})$,
subject to the constraints
\begin{equation*}
\left\{
\begin{array}{rcl}
\displaystyle v + w_1^{\xx} x +w_1^{\yy}y
+\sum_{i=2}^{M-1}w_i^{\xx} (x-K_i^{\xx})_+
+\sum_{i=2}^{N-1}w_i^{\yy} (y-K_i^{\yy})_+
+\delta(x)(y-x)
 & \geq  & |y-\Kf x|, 
\quad\text{for all }(x,y)\in\Xx,\\
w_0^{\xx} + w_0^{\yy} - w_1^{\xx} K_1^{\xx} - w_1^{\yy} K_1^{\yy}  & =  & v.
\end{array}
\right.
\end{equation*}
The dual problem here is semi-infinite dimensional since 
the minimisation is performed over the finite-dimensional vectors~$\ww^{\xx}$ and~$\ww^{\yy}$,
but also over the space of continuous and bounded functions~$\delta$ on~$\RR_+$
(it is enough to consider only continuous and bounded functions instead of bounded measurable functions as pointed out in~\cite[Section 1.5]{BHP11}).
The importance of incorporating the martingale conditions into the discretisation is critical. 
This is easily seen in the following example for the primal problem, which also yields an issue for the dual.
Suppose that $S_t$ can take value $0.75$ or $1.25$ each with $50\%$  probability and $S_{t+\tau}$ can take value $0.5$ or $1.5$ each with $50\%$ probability.
Note that $\EE(S_t)=\EE(S_{t+\tau})=1$.
We consider the primal problem.
The constraints $\|\zeta_{i,\cdot}\| = \mu_i$ and $(\langle\zeta_{i,\cdot}, (x_i-\yy)\rangle)_{i} = \mathbf{0}$ 
fully determine the probabilities
$\zeta_{1,1}=\zeta_{2,2}=3/8$ and $\zeta_{1,2}=\zeta_{2,1}=1/8$.
The final constraints $\|\zeta_{\cdot,j}\| = \nu_j$ are only true if $\nu_1=\nu_2=1/2$ or 
$\EE(S_{t+\tau})=1$. Otherwise, there will be no solution to this LP.
This stresses the importance of a consistent no-arbitrage discretisation of the problem.

\subsubsection{Approximation of the dual}
In order to reduce the dual problem~\eqref{eq:SemiInfDual} to a purely finite-dimensional problem, 
we further add a layer of discretisation for the continuous and bounded delta hedges~$\delta$.
Similar to~\cite{PHL11},
fix $M_b\in\NN$ and a finite-dimensional basis $(\phi_i)_{i=1,\ldots,M_b}$ on $C_b(\RR_+)$, 
and let $\ww^b:=(w_{1}^{b}, \ldots, w_{M_b}^{b})$ be a vector in~$\RR^{M_b}$;
define then the discretised hedge~$\widetilde{\delta}:\RR_+\to\RR$ as 
\begin{equation}\label{eq:DiscretisedDelta}
\widetilde{\delta}(x) := \sum_{i=1}^{M_b}w_i^{b} \phi_i(x),
\end{equation}
so that the new (discretised) dual problem now has the following finite-dimensional form:
\begin{equation}\label{eq:FiniteDual}
\overline{\mathcal{D}}_b(\mu,\nu)
 = \min
\left\{ v +w_1^{\xx}  + w_1^{\yy}
+ \sum_{i=2}^{M-1}w_i^{\xx} \widetilde{C}(t,K_i^{\xx})
+ \sum_{i=2}^{N-1}w_i^{\yy}\widetilde{C}(t+\tau,K_i^{\yy}): 
(\ww^{\xx},\ww^{\yy},\ww^b) \in \RR^{M+N+M_b}\right\},
\end{equation}
subject to the constraints
\begin{equation*}
\left\{
\begin{array}{rcl}
\displaystyle v + w_1^{\xx} x +w_1^{\yy}y
 + \sum_{i=2}^{M-1}w_i^{\xx} (x-K_i^{\xx})_+
 + \sum_{i=2}^{N-1}w_i^{\yy} (y-K_i^{\yy})_+
 + \widetilde{\delta}(x)(y-x)
 & \geq  & |y-\Kf x|, 
\quad\text{for all }(x,y)\in\Xx,\\
w_0^{\xx} + w_0^{\yy} - w_1^{\xx} K_1^{\xx} - w_1^{\yy} K_1^{\yy}  & =  & v.
\end{array}
\right.
\end{equation*}

\section{Primal solution for the at-the-money case}\label{sec:primalnumerics}

In~\cite{HK13}, Hobson and Klimmek derived the lower bound optimal martingale transport plan for the at-the-money ($\Kf=1$) forward-start straddle.
Let $\Delta(z):= \int_0^z f_\nu(u)\D u - \int_0^z f_\mu(u)\D u$ for all $z\geq0$;
then Assumption~\ref{assump:munuassump}---crucial in Hobson and Klimmek's analysis---is 
equivalent to~$\Delta$ having a single maximiser~\cite[Lemma 5.1]{CLM14}.
This assumption imposes constraints on the tail behaviour of the difference 
between the two laws~$\mu$ and~$\nu$, and is clearly satisfied, for instance, in the Black-Scholes case.

\subsection{Structure of the transport plan}
The key risk for an at-the-money forward-start straddle is that a long position is equivalent to being short the kurtosis of the conditional distribution of the underlying asset (see for example in the introduction of~\cite{HN08}). 
Therefore to produce the lowest possible price it seems reasonable 
to require a transport plan that maximises the kurtosis of the conditional distribution. 
This is indeed the structure of the solution in~\cite{HK13}.
We leave as much common mass $(\mu\wedge\nu)$ in place and then map the residual mass $\eta$ on $[a,b]$ to the tails of the distribution $\gamma$ via two decreasing functions 
$p:[a,b]\to [0,a]$ and $q:[a,b]\to [b,\infty)$.
Using the martingale condition,
Hobson and Klimmek~\cite{HK13} derive
a system of coupled differential equations for~$(p,q)$:
\begin{equation}\label{eq:ODEpq}
p'(x)= \frac{q(x)-x}{q(x)-p(x)}\frac{f_{\mu}(x)-f_{\nu}(x)}{f_{\mu}(p(x))-f_{\nu}(p(x))},
\qquad
q'(x)= \frac{x-p(x)}{q(x)-p(x)}\frac{f_{\mu}(x)-f_{\nu}(x)}{f_{\mu}(q(x))-f_{\nu}(q(x))},
\end{equation}
with boundary conditions 
\begin{align*}
p(b) & = \inf\left\{x\geq 0: \gamma([0,x]) > 0\right\},\\
q(b) & = \inf\left\{x\geq 0: \gamma([0,x]) > \gamma([0,b])\right\},\\
p(a) & = \sup\left\{x\geq 0: \eta([0,x]) <\eta([0,a])\right\},\\
q(a) & = \sup\left\{x\geq 0: \gamma([0,x]) <1\right\}.
\end{align*}
By taking limits, we obtain that $p(a)=a$, $p(b)=0$
$q(a) = +\infty$ and $q(b) = b$, see also~\cite[Proposition 5.6]{CLM14}.

\subsection{Implementation}
The right-hand side of the equations in~\eqref{eq:ODEpq} are undefined at the boundary points.
An application of L'H\^opital's rule shows that $\lim_{x\uparrow b}q'(x)=-1$ if
$f_{\mu}'(b)\neq f_{\nu}'(b)$, which is a reasonable assumption in practice, as will be illustrated in Section~\ref{sec:OTnumerics}.
On the other hand $\lim_{x\uparrow b}p'(x)$ depends on the marginal measures~$\mu$ and~$\nu$.
For instance, in the lognormal example in Section~\ref{sec:OTnumerics}, we find that $p'(x)=\mathcal{O}\left(\E^{\alpha(\log p(x))^2}\right)$ for some $\alpha>0$
as $x$ tends to~$b$ from below, 
and $\lim_{x\uparrow b}p'(x)=-\infty$ (see for example Figure~\ref{fig:BSPDF}(b)).
On the other hand if for example $f_{\mu}(0)\neq f_{\nu}(0)$ or  $f_{\mu}'(0)\neq f_{\nu}'(0)$ then $\lim_{x\uparrow b}p'(x)=0$. 

In order to circumvent these issues so that we can apply the Runge-Kutta method to solve these ODEs,
we introduce the following pre-processing step: 
fix a small $\eps>0$ (in our implementations, we choose $\eps=0.001$),
integrate both sides of~\eqref{eq:ODEpq} over $[b-\eps,b]$ 
and then approximate the right-hand side by using the rectangle rule for the integral with the unknown values
$p^*:=p(b-\eps)$ and  $q^*:=q(b-\eps)$.
This yields the following simultaneous equations for $p^*$ and $q^*$ which we solve numerically:
\begin{align*}
p^*=-\eps \frac{q^*-b+\eps}{q^*-p^*}\frac{f_{\mu}(b-\eps)-f_{\nu}(b-\eps)}{f_{\mu}(p^*)-f_{\nu}(p^*)}, \\ 
q^*=b - \eps \frac{b-\eps-p^*}{q^*-p^*}\frac{f_{\mu}(b-\eps)-f_{\nu}(b-\eps)}{f_{\mu}(q^*)-f_{\nu}(q^*)}.
\end{align*}
These equations can easily be reduced into one root search; the first equation gives 
$$
q^*=\frac{(p^*)^2\left[f_{\mu}(p^*)-f_{\nu}(p^*)\right]+\eps (b-\eps)\left[f_{\mu}(b-\eps)-f_{\nu}(b-\eps)\right] }
{p^*\left[f_{\mu}(p^*)-f_{\nu}(p^*)\right]+\eps \left[f_{\mu}(b-\eps)-f_{\nu}(b-\eps)\right]},
$$
which in turn can be plugged into the second equation to solve for $p^*$. The pair $(p,q)$ in~\eqref{eq:ODEpq} is then solved for $x\in[a,b-\eps]$ using standard Runge-Kutta methods with the new boundary conditions $p(b-\eps)=p^*, q(b-\eps)=q^*$.
The lower bound for the at-the-money forward-start straddle price is then given using the optimal
conditional density $\rho^*(Y=y\vert X=x)$ given in~\cite[page 199]{HK13}:
\begin{equation*}
\rho^*(Y=y\vert X=x) = 
\left\{
\begin{array}{ll}
\displaystyle \frac{f_{\eta}(x)}{f_{\mu}(x)} \frac{q(x)-x}{q(x)-p(x)}\ind_{\{y=p(x)\}}, & \text{if } y<x,\\
\displaystyle \left(1-\frac{f_{\eta}(x)}{f_{\mu}(x)}\right), & \text{if } y=x,\\
\displaystyle \frac{f_{\eta}(x) }{f_{\mu}(x)}\frac{x-p(x)}{q(x)-p(x)}\ind_{\{y=q(x)\}}, & \text{if } y>x,
\end{array}
\right.
\end{equation*}
and hence straightforward computations yield
\begin{align}\label{eq:HKSolution}
\mathbb{E}(|Y-X|)
 & = \int_{\RR}\mathbb{E}(|Y-X|\vert X=x)f_{\mu}(x) \D x
 = \int_{\RR}\int_{\RR}\rho^*(Y=y \vert X=x)f_{\mu}(x)|y-x|\D y \D x\nonumber \\
 &  = \int_{\RR}\int_{-\infty}^{x}f_{\eta}(x) \frac{q(x)-x}{q(x)-p(x)}\ind_{\{y=p(x)\}}|y-x|  \D y\D x
 + \int_{\RR}\int_{x}^{+\infty}f_{\eta}(x) \frac{x-p(x)}{q(x)-p(x)}\ind_{\{y=q(x)\}}|y-x|  \D y\D x\nonumber\\
 &  = \int_{a}^{b}\frac{2(x-p(x))(q(x)-x)}{q(x)-p(x)} f_{\eta}(x) \D x.
\end{align}

\section{Numerical analysis of the no-arbitrage bounds}\label{sec:OTnumerics}
We now illustrate the numerical methods developed in Sections~\ref{sec:noarbdiscretopttrans} and~\ref{sec:primalnumerics} on the Black-Scholes and the Heston models.
These examples involve forward-start option prices (and forward implied volatility smile), 
which we now quickly recall.
In the Black-Scholes model, the dynamics of the stock price process under the risk-neutral measure are given by
$\D S_t=S_t \Sigma \D W_t$, $S_0=1$,
where $\Sigma>0$ represents the instantaneous volatility and $W$ is a standard Brownian motion.
The no-arbitrage price of the Call option at time zero then reads
$\BS(\tau,K,\Sigma) := \mathbb{E}\left(S_\tau-K\right)^+
=\Nn\left(d_+\right)-\E^k\Nn\left(d_-\right)$,
with $d_\pm:=-\frac{k}{\Sigma\sqrt{\tau}}\pm\frac{1}{2}\Sigma\sqrt{\tau}$,
where $\Nn$ is the standard normal distribution function. 
Since the increments of the stock price process are stationary and independent, 
the forward-start option with  payoff $(S_{t+\tau}-KS_t)^+$ with $t,\tau>0$
is worth $\BS(\tau,K,\Sigma)$. 
For a given market or model price $C^{\textrm{obs}}(t,\tau,K)$ of the option at strike~$K$, 
forward-start date~$t$ and maturity~$\tau$,
the forward implied volatility smile $\sigma_{t,\tau}(K)$ is then defined
 as the unique solution to 
$C^{\textrm{obs}}(t,\tau,K) = \BS(\tau,k,\sigma_{t,\tau}(K))$. 

In the following two subsections, we shall consider the observed vanilla Call option prices
as computed from the Black-Scholes and the Heston model.
We shall discretise the corresponding dual problems following Section~\ref{sec:DualDiscretisation},
by considering the compact interval~$[0,5]$
for the supports of the discretised random variables~$\St_t$ and~$\St_{t+\tau}$ with $m=n=500$ discretisation points;
the vectors of observed strikes are taken as 
$\Km^{\xx} = \Km^{\yy} = \left\{0.3, 0.4, 0.5,\ldots,2\right\}$.

\subsection{Application to the Black-Scholes model}\label{subsection:BlackScholesApp}
Let $\mathcal{N}(\overline{m},\Sigma^2)$ denote the Gaussian distribution with mean~$\overline{m}$ 
and variance~$\Sigma^2$,
and assume that the random variables~$S_t$ and~$S_{t+\tau}$ are distributed according to
$$
\log(S_t)\sim\mathcal{N}\left(-\frac{1}{2}\Sigma^2 t,\Sigma^2 t\right)
\qquad\text{and}\qquad
\log(S_{t+\tau})\sim\mathcal{N}\left(-\frac{1}{2}\Sigma^2 (t+\tau),\Sigma^2 (t+\tau)\right).
$$ 
Clearly a candidate martingale coupling  is the Black-Scholes model with volatility~$\Sigma$
and starting at $S_0 = 1$;
in this case the forward volatility, i.e. the implied volatility computed from the forward-start option,
 is constant and also equal to~$\Sigma$.
In Figure~\ref{fig:BSPDF}(a),
we consider the values $\Sigma=0.2$, $t=1$ and $\tau=0.5$, 
and plot the distributions of $S_t$ and $S_{t+\tau}$.

In Figure~\ref{fig:BSPDF}(b) we plot the lower and upper bounds for the forward implied volatility smile
computed from the discretisation of the dual problems~\eqref{eq:dual}
via~\eqref{eq:FiniteDual}.
The lower bound at-the-money case using the Hobson-Klimmek solution and the LP dual solution are virtually identical ($6.95\%$ vs $6.98\%$), illustrating the consistency of the two approaches.
Note that even in this simple case the range of possible forward smiles consistent with the two marginal laws is wide.
The magnitude of the bounds produced is similar to bounds obtained in \cite{PHL11}
for cliquet options which are very closely related to forward-start straddles.
This behaviour can be explained intuitively as no conditional instruments are used to hedge the straddle, 
hence there is no restriction on the conditional probabilities between times $t$ and $t+\tau$. 
The only restriction is that the two marginal laws at those times are placed in convex order, 
producing a very large class of feasible martingale measures.

\begin{figure}[h!tb] 
\centering
\mbox{\subfigure[Densities.]{\includegraphics[scale=0.8]{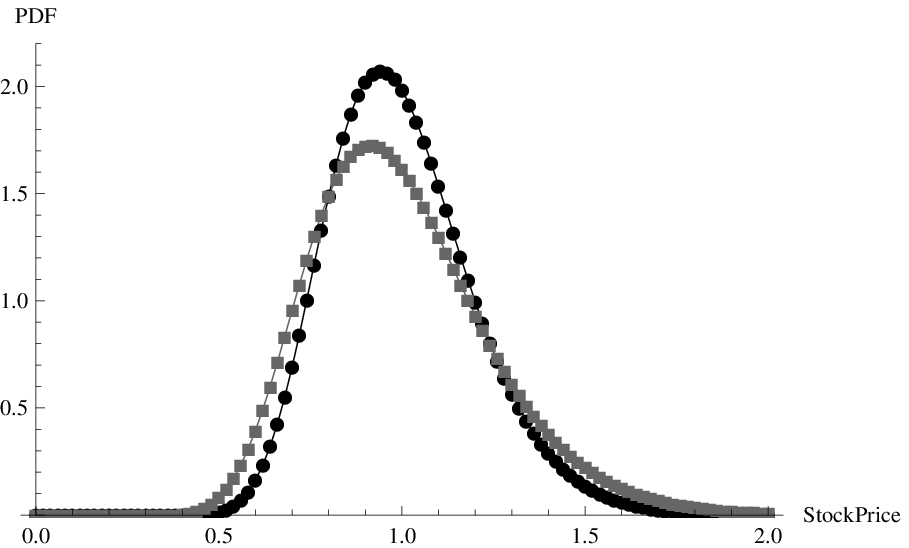}}\quad
\subfigure[Robust bounds via the dual problem.]{\includegraphics[scale=0.8]{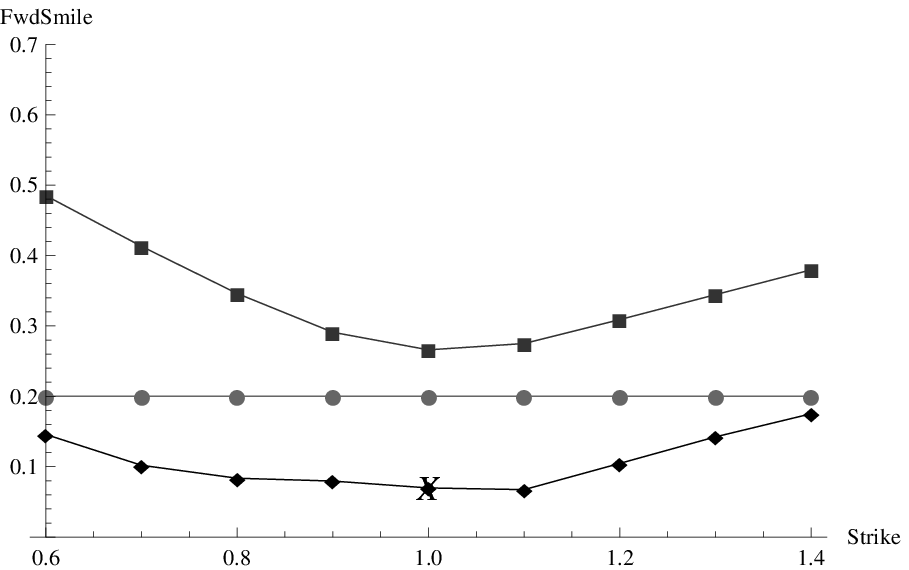}}}
\caption{(a) the circles represent the one-year lognormal density and the squares the $1.5$-year lognormal density.
(b) the circles represent the constant Black-Scholes forward volatility~$\Sigma$ consistent with the marginals.
The squares and the diamonds are the lower and upper bounds found by solving the dual problems~\eqref{eq:dual} 
via~\eqref{eq:FiniteDual};
the X cross is the Hobson-Klimmek solution~\eqref{eq:HKSolution} for the lower bound at-the-money case.
}
\label{fig:BSPDF}
\end{figure}

\subsection{Application to the Heston model}\label{subsection:HestonApp}
The marginal distributions for expiries $t=1$ and $t+\tau=1.5$ are now generated according to 
the Heston stochastic volatility model~\cite{Heston}, 
in which the stock price process is the unique strong solution to the stochastic differential equation
\begin{equation}
\begin{array}{rcll}
\D S_t & = & S_t \sqrt{V_t} \D W_t, & S_0 = 1,\\
\D V_t & = & \kappa\left(\theta-V_t\right)\D t + \xi\sqrt{V_t} \D Z_t, & V_0 = v>0,
\end{array}
\end{equation}
where $W$ and $Z$ are two one-dimensional standard Brownian motions with
$\D\langle W, Z\rangle_t = \rho \D t$, and $\kappa, \theta, \xi>0$ and $\rho \in [-1,1]$.
We consider here the following values: $v=\theta=0.07$, $\kappa=1$, $\xi=0.4$ and $\rho=-0.8$.
The (spot) implied volatility smiles and corresponding densities are displayed in Figure~\ref{fig:HestonPDF}.
Figure~\ref{fig:HestonBound} shows the Heston forward smile consistent with the marginals 
(computed using the inverse Fourier transform representation and a simple root search method) and the lower and upper bounds for the forward smile, 
from the discretised version~\eqref{eq:FiniteDual} of the dual problems~\eqref{eq:dual}.
For the discretisation of the delta hedge~\eqref{eq:DiscretisedDelta}, 
we consider the monomials $(\phi_1, \phi_2, \phi_3)(x) \equiv (1, x, x^2)$.
As in the Black-Scholes case, the Hobson-Klimmek solution ($7.77\%$) and the dual solution ($7.80\%$) 
for the lower-bound at-the-money volatility are virtually identical.
Figure~\ref{fig:CallDeltaHedge}(a) shows the payoff of the option prices in the super-hedge:
one enters into positions with long convexity for the $1.5$-year maturity and short convexity 
for the $1$-year maturity.


\begin{figure}[h!tb] 
\centering
\mbox{
\subfigure[Densities.]{\includegraphics[scale=0.8]{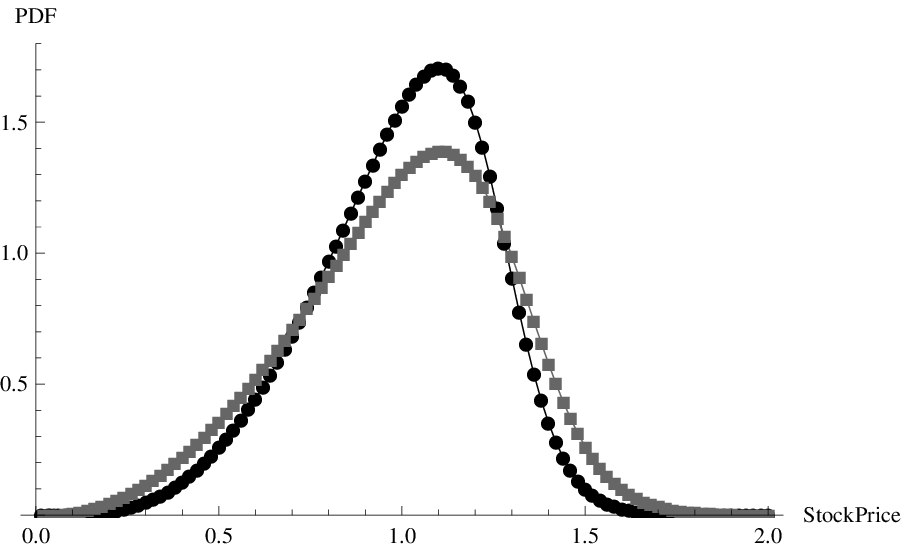}}\quad
\subfigure[Spot implied volatility.]{\includegraphics[scale=0.8]{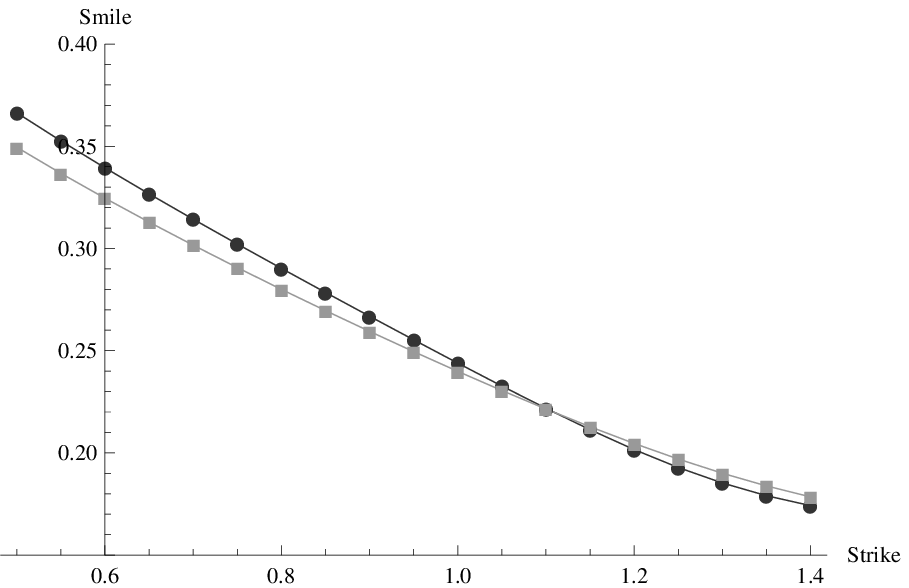}}
}
\caption{(a) Circles (squares) represents the $1$ year ($1.5$ year) marginal densities.
(b) Circles (squares) represent the corresponding spot implied volatilities.}
\label{fig:HestonPDF}
\end{figure}

\begin{figure}[h!tb] 
\centering
\includegraphics[scale=0.8]{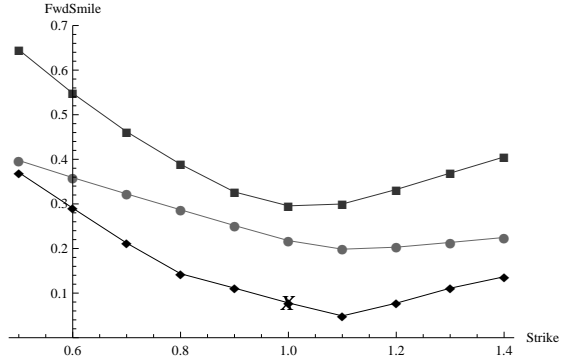}
\caption{The circles represent the Heston forward volatility consistent with the marginals, 
and the squares and diamonds stand far the lower and upper bounds found by solving the LP problem in Section~\ref{sec:noarbdiscretopttrans},
and X is the primal Hobson-Klimmek solution for the lower bound at-the-money case (Section~\ref{sec:primalnumerics}).}
\label{fig:HestonBound}
\end{figure}

\begin{figure}[h!tb] 
\centering
\mbox{\subfigure[Super-hedging portfolio and forward-start payoffs]{\includegraphics[scale=0.8]{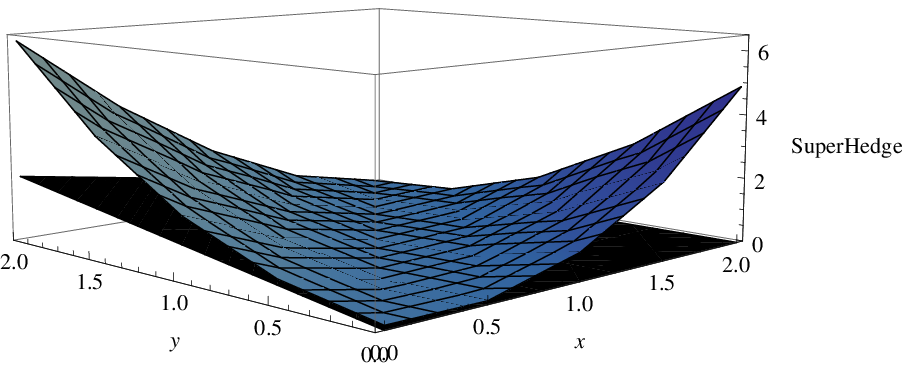}}\quad
\subfigure[Discretised Delta hedge]{\includegraphics[scale=0.8]{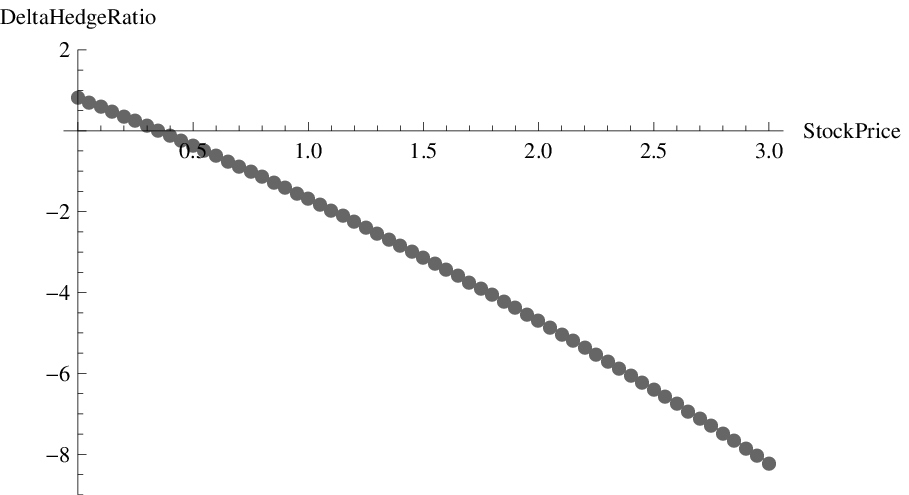}}}
\caption{The strike is taken at the money: ~$\Kf=1$.
(a) Super-hedging portfolio (top plot) and forward-start payoffs (below);
(b) Optimal discretised delta hedge $\widetilde{\delta}$ defined in~\eqref{eq:DiscretisedDelta}.}
\label{fig:CallDeltaHedge}
\end{figure}

In both examples the range of forward smiles consistent with the marginal laws is large.
Using European options to `lock-in' (replicate) forward volatility or hedge forward volatility dependent claims
seems illusory.
Forward-start options should be seen as fundamental building blocks for exotic pricing and not decomposable (or approximately decomposable) into European options.
Models used for forward volatility dependent exotics should have the capability of calibration to forward-start option prices and at a minimum should produce realistic forward smiles that are consistent with trader expectations and observable prices.

\subsection{A note on the discretisation methodology}
In both examples above, the discretised random variables~$\St_{t}$ and~$\St_{t+\tau}$ 
were supported on~$500$ points. 
This choice was arbitrary, and it is natural to question it.
In the semi-infinite case, the absence of duality gap between the primal problem and its dual 
guarantees~\cite[Theorem 3.1]{Shapiro09} the existence of a discretisation, the value of which converges to that of the primal problem as the number of points increases.
Rates of convergence have also been obtained for discretisation schemes 
in semi-infinite programming~\cite{Shapiro09, Still}.
Our setting here (Primal problem~\eqref{eq:primal} and its dual~\eqref{eq:dual}) is however
of infinite-dimensional nature, and, to the best of our knowledge, no corresponding result exists (yet!).
A general, theoretical, proof is outside the scope of our approach, and we now provide some numerical evidence
about the convergence and stability of our discretisation scheme~\eqref{eq:FiniteDual}. 
We choose two different discretisation grids 
for the supports of the random variables~$\St_{t}$ and~$\St_{t+\tau}$, 
following Algorithm~\ref{algo:generatediscrete}: 
a uniform grid and a grid consisting of roots of Legendre polynomials, both with the same number of points.
Below, Tables~\ref{table:subhedgegauss} and~\ref{table:subhedgeuniform} represent the optimal values of the sub-hedging dual problem as the number of discretisation points increases, 
for different values of the forward-start straddle strike~$\Kf$. 
Clearly, refining the discretisation grid produces only minor changes in the optimal values for the at-the-money case ($\Kf=1$) and virtually none for the other cases. 
Tables~\ref{table:superhedgegauss} and~\ref{table:superhedgeuniform} represent the optimal values of the super-hedging dual problem for different numbers of discretisation points and different values of the forward-start straddle strike~$\Kf$.
Refining the discretisation grid yields only minor changes to the optimal values. 
In contrast to the sub-hedging case, the change in optimal values for the super-hedging problem 
can be observed for all strikes.
Finally, both sub- and super-hedging dual problems seem to be very stable with respect to the discretistation schemes and the number of points. 

We would also like to comment on the rate of convergence of the optimal portfolios with respect to the partition refinements. 
Hobson and Neurberger~\cite[Section 6.2]{HN08} obtained
explicit expressions for the dual variables~$(\psi_0, \psi_1,\delta)$, defined in~\eqref{eq:dual},
in the lognormal case for the at-the-money forward-start straddle ($\Kf = 1$): 
\begin{equation}\label{eq:HN_lognormal_sln}
\begin{array}{ll}
&\psi_0(x)  = -\xi x\ln x + \xi x\ln\left(A\sinh (\xi^{-1})\right) + x\coth(\xi^{-1}),\\
&\psi_1(y)  = \xi\left(y\ln y - y\ln(A/\xi) - y\right), \\
&\delta(x)  = -\xi \ln\left(x / A\sinh(\xi^{-1})\right),
\end{array}
\end{equation}
where the constants $A$ and $\xi$ are such that the expected value of the hedging portfolio 
$\psi_0(x) + \psi_1(y) + \delta(x)(y-x)$ is minimised under super-hedging constraints.
They also showed that the bound achieved by the solution is tight~\cite[Section 10]{HN08}.
We compute the sup-norm error $\eps(n)$ ($n$ is the number of discretisation points) between solutions of the discretised Dual~\eqref{eq:FiniteDual} and the Hobson-Neuberger solution~\eqref{eq:HN_lognormal_sln}.
We visually check the rate of convergence (i.e. the highest exponent $r$ such that $\eps(n) \sim \Oo(d_n^r)$)
with respect to the discretisation mesh size~$d_n$ by plotting
$\log(\eps(n))/\log(d_n)$ against $\log(d_n)$. 
The resulting plot is presented in Figure~\ref{fig:ConvergenceLog} below. 
As mentioned above, no theoretical rates of convergence exist for infinite-dimensional linear programming problems;
in the semi-infinite case, the corresponding plot would be roughly constant.

\begin{table}[h!tb]
\centering
	\begin{tabular}{| c | c | c | c | c | c |}
	\hline
	\diagbox{Strike}{Number of points} & 75 & 250 & 500 & 1000 & 2000 \\ \hline 
	0.6 & 0.4 & 0.4 & 0.4 & 0.4 & 0.4 \\ \hline
	0.7 & 0.3 & 0.3 & 0.3 & 0.3 & 0.3 \\ \hline
	0.8 & 0.2 & 0.2 & 0.2 & 0.2 & 0.2 \\ \hline
	0.9 & 0.1001 & 0.1 & 0.1 & 0.1 & 0.1 \\ \hline
	1.0 & 0.0390 & 0.0384 & 0.0384 & 0.0384 & 0.0383 \\ \hline
	1.1 & 0.1004 & 0.1004 & 0.1004 & 0.1004 & 0.1004 \\ \hline
	1.2 & 0.2 & 0.2 & 0.2 & 0.2 & 0.2 \\ \hline
	1.3 & 0.3 & 0.3 & 0.3 & 0.3 & 0.3 \\ \hline
	1.4 & 0.4 & 0.4 & 0.4 & 0.4 & 0.4 \\ \hline
	\end{tabular}
\caption{Optimal values of the sub-hedging dual problem (using roots of Legendre polynomials).}
\label{table:subhedgegauss}	
\end{table}

\begin{table}[h!tb]
\centering
	\begin{tabular}{| c | c | c | c | c | c |}
	\hline
	\diagbox{Strike}{Number of points} & 75 & 250 & 500 & 1000 & 2000 \\ \hline 
	0.6 & 0.4 & 0.4 & 0.4 & 0.4 & 0.4 \\ \hline
	0.7 & 0.3 & 0.3 & 0.3 & 0.3 & 0.3 \\ \hline
	0.8 & 0.2 & 0.2 & 0.2 & 0.2 & 0.2 \\ \hline
	0.9 & 0.10008 & 0.1 & 0.1 & 0.1 & 0.1 \\ \hline
	1.0 & 0.0388 & 0.0384 & 0.0384 & 0.0384 & 0.0383 \\ \hline
	1.1 & 0.1006 & 0.1004 & 0.1004 & 0.1004 & 0.1004 \\ \hline
	1.2 & 0.2 & 0.2 & 0.2 & 0.2 & 0.2 \\ \hline
	1.3 & 0.3 & 0.3 & 0.3 & 0.3 & 0.3 \\ \hline
	1.4 & 0.4 & 0.4 & 0.4 & 0.4 & 0.4 \\ \hline
	\end{tabular}
\caption{Optimal values of the sub-hedging dual problem (using a uniform grid).}
\label{table:subhedgeuniform}	
\end{table}

\begin{table}[h!tb]
\centering
	\begin{tabular}{| c | c | c | c | c | c |}
	\hline
	\diagbox{Strike}{Number of points} & 75 & 250 & 500 & 1000 & 2000 \\ \hline 
	0.6 & 0.4147  & 0.4156 & 0.4157 & 0.4157 & 0.4157 \\ \hline
	0.7 & 0.3241 & 0.3255 & 0.3257 & 0.3257 & 0.3257 \\ \hline
	0.8 & 0.2394 & 0.2411 & 0.2413 &	0.2413 & 0.2414 \\ \hline
	0.9 & 0.1707 & 0.1743 & 0.1745 & 0.1746 & 0.1746 \\ \hline
	1.0 & 0.1453 & 0.1485 & 0.1489 & 0.1490 & 0.1490 \\ \hline
	1.1 & 0.1786 & 0.1815 & 0.1816 & 0.1817 & 0.1817 \\ \hline
	1.2 & 0.2513 & 0.2536 & 0.2538 & 0.2539 & 0.2539 \\ \hline
	1.3 & 0.3373 & 0.3394 & 0.3396 & 0.3397 & 0.3397 \\ \hline
	1.4 & 0.4292 & 0.4314 & 0.4316 & 0.4316 & 0.4316 \\ \hline
	\end{tabular}
\caption{Optimal values of the super-hedging dual problem (using roots of Legendre polynomials).}
\label{table:superhedgegauss}	
\end{table}

\begin{table}[h!tb]
\centering
	\begin{tabular}{| c | c | c | c | c | c |}
	\hline
	\diagbox{Strike}{Number of points} & 75 & 250 & 500 & 1000 & 2000 \\ \hline 
	0.6 & 0.4150 & 0.4157 & 0.4157 & 0.4157 & 0.4157 \\ \hline
	0.7 & 0.3245 & 0.3256 & 0.3257 & 0.3257 & 0.3257 \\ \hline
	0.8 & 0.2397 & 0.2411 & 0.2413 & 0.2413 & 0.2414 \\ \hline
	0.9 & 0.1716 & 0.1743 & 0.1746 & 0.1746 & 0.1746 \\ \hline
	1.0 & 0.1463 & 0.1487 & 0.1489 & 0.1490 & 0.1490 \\ \hline
	1.1 & 0.1795 & 0.1814 & 0.1817 & 0.1817 & 0.1817 \\ \hline
	1.2 & 0.2514 & 0.2538 & 0.2539 & 0.2539 & 0.2539 \\ \hline
	1.3 & 0.3376 & 0.3395 & 0.3396 & 0.3397 & 0.3397 \\ \hline
	1.4 & 0.4300 & 0.4315 & 0.4316 & 0.4316 & 0.4316 \\ \hline
	\end{tabular}
\caption{Optimal values of the super-hedging dual problem (using a uniform grid).}
\label{table:superhedgeuniform}	
\end{table}

\begin{figure}[h!tb] 
\centering
\includegraphics[scale=0.6]{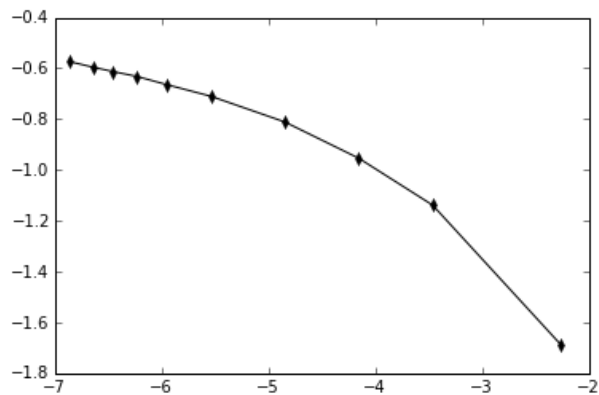}
\caption{Plot of $\log(\eps(n))/\log(d_n)$ against $\log(d_n)$, where $\eps(n)$ is the sup-norm error.}
\label{fig:ConvergenceLog}
\end{figure}

\newpage

\section{Numerical analysis of the transport plans}\label{sec:numericTPs}

As mentioned in Section~\ref{sec:primalnumerics}, the key risk for the at-the-money forward-start straddle is that a long position is equivalent to being short the kurtosis of the conditional distribution.
The solution in the lower bound case (under Assumption~\ref{assump:munuassump}) was detailed in Section~\ref{sec:primalnumerics}, 
where -- intuitively -- the transport plan maximises the kurtosis of the conditional distribution.
In the upper bound case (see~\cite{HN08}) the support of the transport plan is concentrated on a binomial map with no mass being left in place, 
i.e. all the mass of~$\mu$ gets mapped to~$\nu$ via two increasing, continuous and differentiable functions
$f,g:\RR_+\to\RR_{+}$ satisfying $f(x)\leq x \leq g(x)$. 
Functions $f$ and $g$ must satisfy the system of integral equations \cite[Equations~(5.19) and (5.20)]{HN08}
\begin{equation*} 
\left\{
\begin{array}{ll}
\displaystyle 0 = \int_{g^{-1}(y)}^{f^{-1}(y)}{\frac{\left(g(z) - f^{-1}(y)\right)}{g(z) - f(z)}\D z},\\
\displaystyle 1 = \int_{g^{-1}(y)}^{f^{-1}(y)}{\frac{1}{g(z) - f(z)}\D z}, 
\end{array}
\right.
\end{equation*}
and if it is possible to find a solution to this system then \cite[Lemma~7.1]{HN08} provides an optimality result in the at-the-money case.
Intuitively in this case the solution minimises the kurtosis of the conditional distribution.
For out-of-the-money options the situation is more subtle.
As the strike moves further away from the money, a long option position becomes longer the kurtosis of the conditional distribution.
Intuitively one would then expect the transport plan to be some combination of the lower and upper at-the-money transport plans discussed above.
Using the lognormal example of Section~\ref{sec:OTnumerics}, 
we now numerically solve for the transport plans using the discretised primal problem~\eqref{eq:DiscretisedPrimal}
and make qualitative conjectures concerning the structure of the transport plans.
The supports of the discretised random variables~$\widetilde{S}_t$ and~$\widetilde{S}_{t+\tau}$
(introduced in Section~\ref{sec:DualDiscretisation})
are taken respectively as~$[0,10]$ (with $m=1000$ points) and~$[0,30]$ (with $n=3000$ points);
the vectors of observed strikes are again
$\Km^{\xx} = \Km^{\yy} = \left\{0.3, 0.4, 0.5,\ldots,2\right\}$,
and the discrete probabilities in each bucket were obtained by integrating the lognormal density over each bucket.

\begin{figure}[h!tb] 
\centering
\includegraphics[scale=0.6]{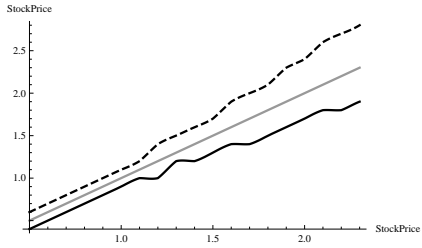}
\caption{The dashed and dark lines are the transport maps for the upper bound at-the-money case
and the grey line is the identity.
The horizontal and vertical axes are $S_t$ and $S_{t+\tau}$.  }
\label{fig:UpperBoundATM}
\end{figure}

\begin{figure}[h!tb] 
\centering
\mbox{\subfigure[Lower Bound.]{\includegraphics[scale=0.6]{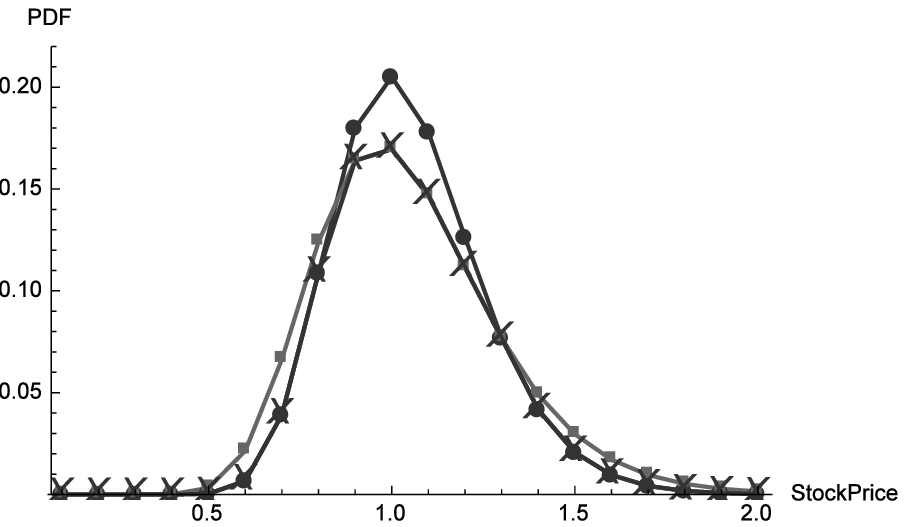}}}
\mbox{\subfigure[Transport maps]{\includegraphics[scale=0.6]{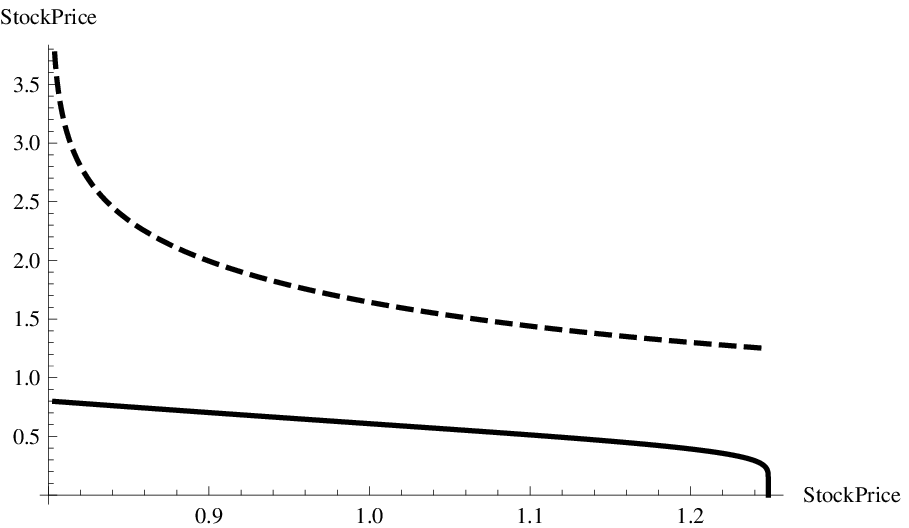}}}
\caption{Here the strike is taken at the money $\Kf=1$.
(a) Discretisation of the measures~$\mu$ (circles) and~$\nu$ (squares) and the amount of mass that must be left in place (X's) in the transport plan for the at-the-money ($\Kf=1$) lower bound case.
(b) Transport maps for the residual mass, computed from~\eqref{eq:ODEpq} or, equivalently, 
by solving the primal problem~\eqref{eq:DiscretisedPrimal}.
}
\label{fig:LowerBoundATM}
\end{figure}

\begin{figure}[h!tb] 
\centering
\mbox{\subfigure[Mass in place $\Kf=0.9$: Upper Bound.]{\includegraphics[scale=0.6]{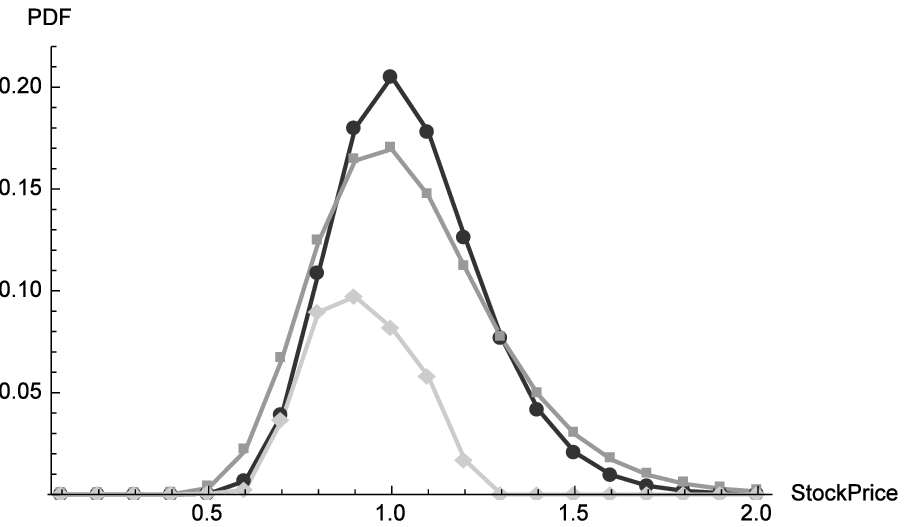}}\quad
\subfigure[Transport Maps $\Kf=0.9$: Upper Bound.]{\includegraphics[scale=0.6]{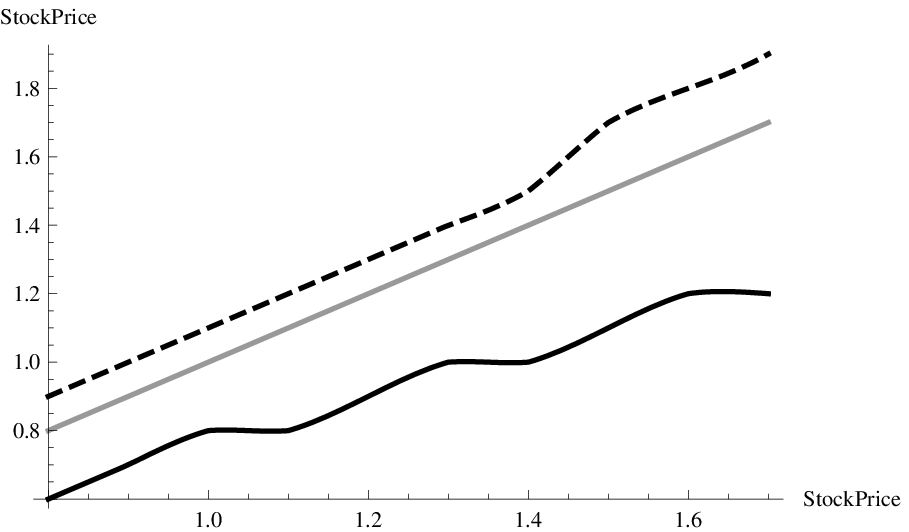}}}
\caption{(a) Discretisation of the measures $\mu$ (circles), $\nu$ (squares) 
and the amount of mass that must be left in place (diamonds) in the transport plan for the $\Kf=0.9$ upper bound case.
(b) Transport maps for the residual mass: the axes are labelled as in Figure~\ref{fig:UpperBoundATM}. }
\label{fig:TransportUB0.9}
\end{figure}

\begin{figure}[h!tb] 
\centering
\mbox{\subfigure[Mass in place $\Kf=0.7$: Upper Bound.]{\includegraphics[scale=0.6]{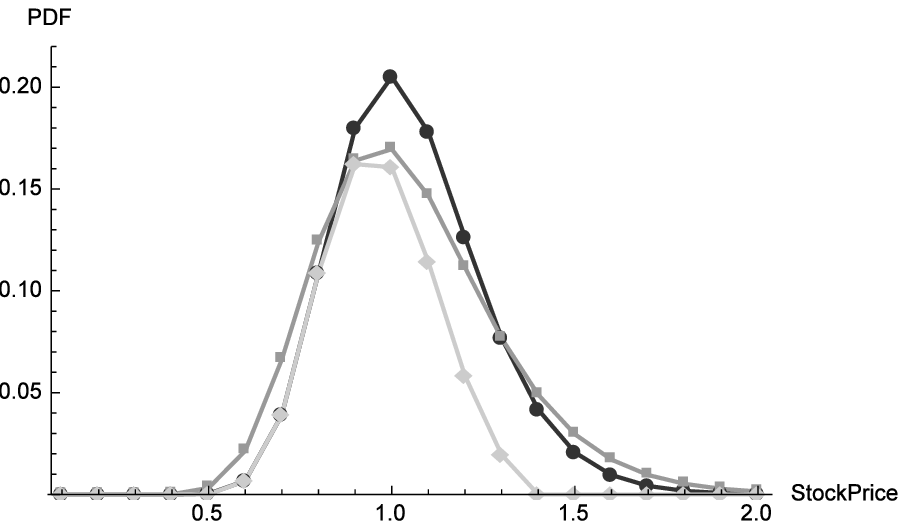}}\quad
\subfigure[Transport Maps $\Kf=0.7$: Upper Bound.]{\includegraphics[scale=0.6]{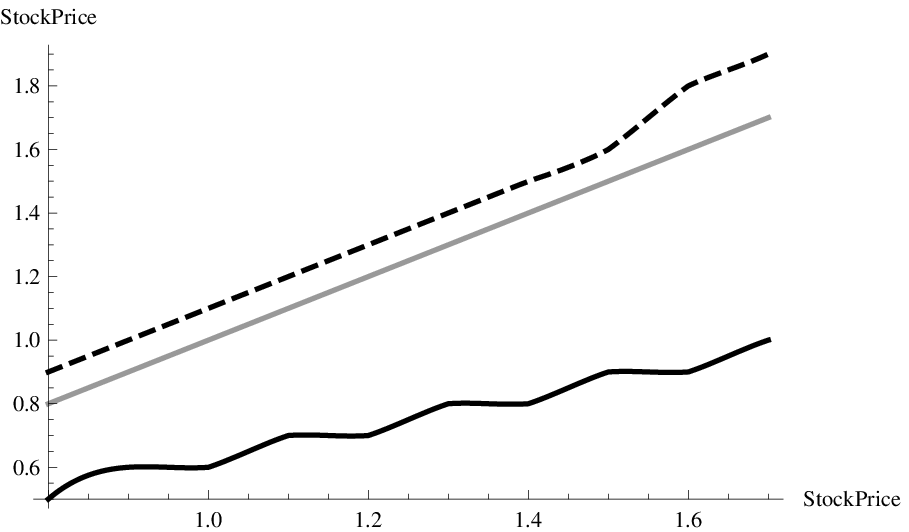}}}
\caption{(a) Discretisation of the measures $\mu$ (circles), $\nu$ (squares) 
and the amount of mass that must be left in place (diamonds) in the transport plan for the $\Kf=0.7$ upper bound case.
(b) Transport maps for the residual mass: the axes are labelled as in Figure~\ref{fig:UpperBoundATM}. }
\label{fig:TransportUB0.7}
\end{figure}

\begin{figure}[h!tb] 
\centering
\mbox{\subfigure[Mass in place $\Kf=1.05$: Lower Bound.]{\includegraphics[scale=0.6]{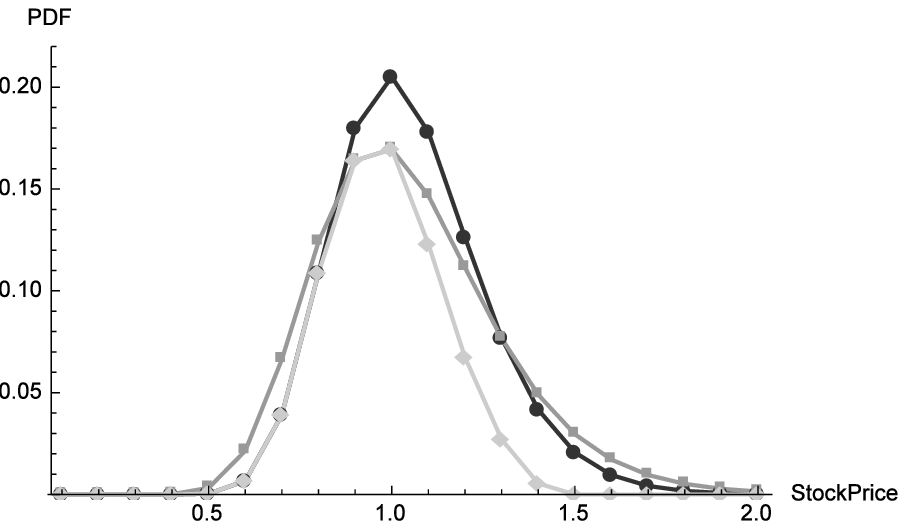}}\quad
\subfigure[Transport Maps $\Kf=1.05$: Lower Bound.]{\includegraphics[scale=0.6]{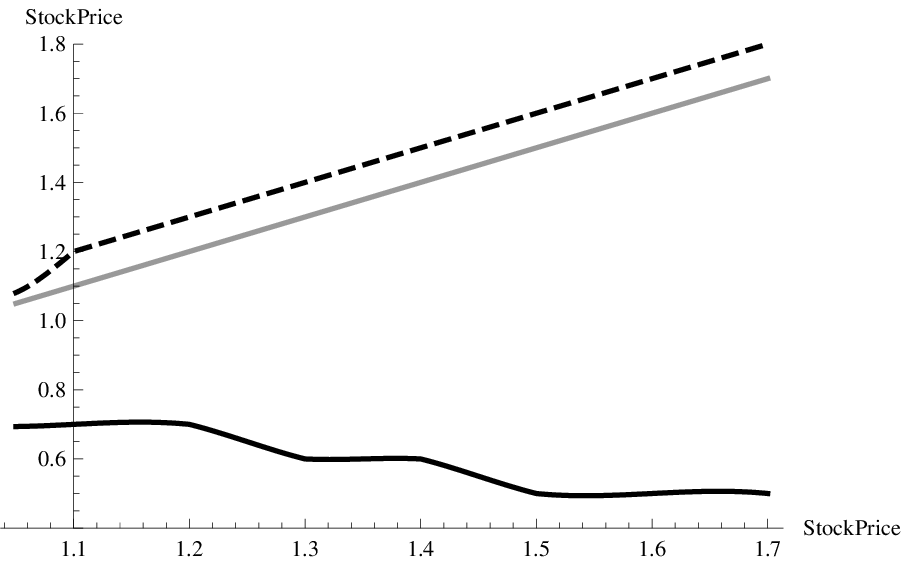}}}
\caption{(a) Discretisation of the measures $\mu$ (circles), $\nu$ (squares) 
and the amount of mass that must be left in place (diamonds) in the transport plan for the $\Kf=1.05$ lower bound case.
(b) Transport maps for the residual mass: the axes are labelled as in Figure~\ref{fig:UpperBoundATM}. }
\label{fig:TransportLB1.05}
\end{figure}

\begin{figure}[h!tb] 
\centering
\mbox{\subfigure[Mass in place $\Kf=1.3$: lower bound.]{\includegraphics[scale=0.6]{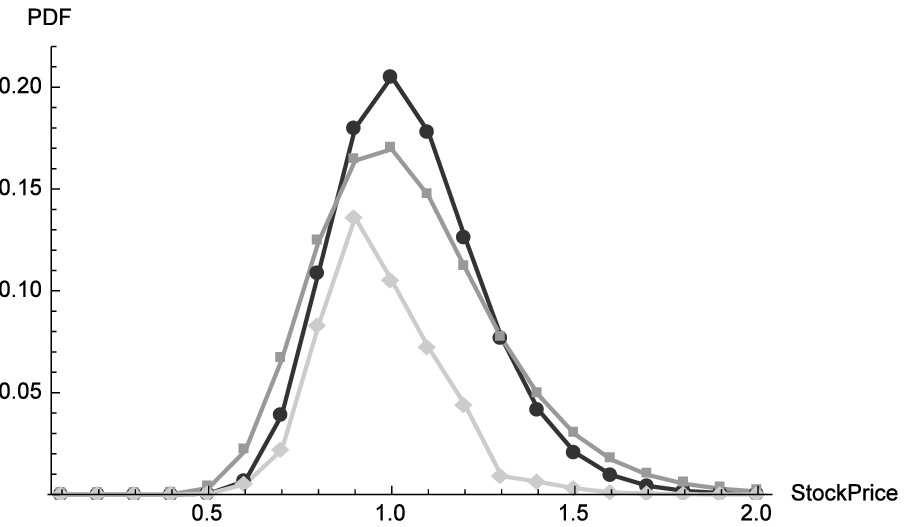}}\quad
\subfigure[Transport maps $\Kf=1.3$: lower bound.]{\includegraphics[scale=0.6]{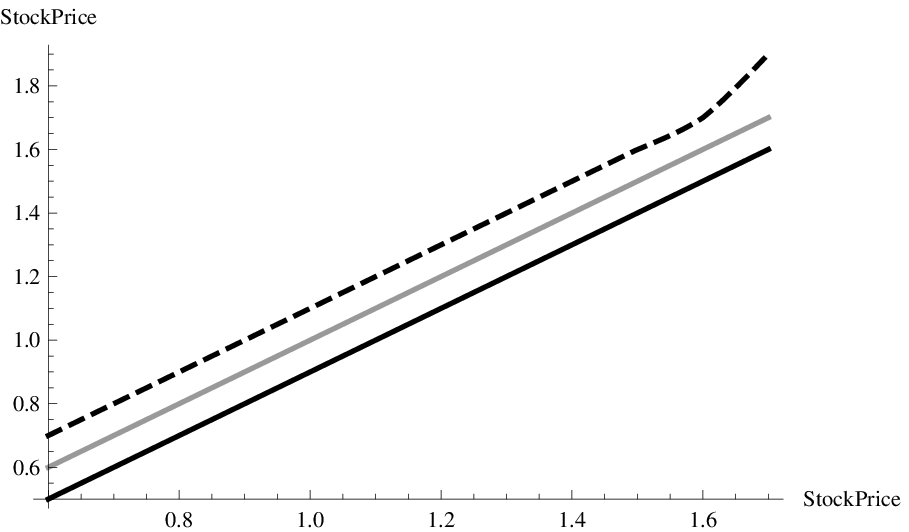}}}
\caption{(a) Discretisation of the measures $\mu$ (circles), $\nu$ (squares) 
and the amount of mass left in place (diamonds) in the transport plan for the $\Kf=1.3$ lower bound case.
(b) Transport maps for the residual mass: the axes are labelled as in Figure~\ref{fig:UpperBoundATM}. }
\label{fig:TransportLB1.3}
\end{figure}

In Figure~\ref{fig:UpperBoundATM} we compute the transport maps~$f$
and~$g$ for the at-the-money upper bound case, i.e. the supremum case in~\eqref{eq:primal}.
In this case no mass is left in place in the transport plan;
in Figure~\ref{fig:LowerBoundATM} we plot the transport plan for the at-the-money lower bound case.
This lower bound case is in striking agreement with Hobson-Klimmek: 
as much mass as possible is left in place and the residual mass is mapped to the tails of the distribution via two decreasing functions.
Note the agreement with the transport maps in Figure~\ref{fig:BSPDF}(b).
In this case the forward volatility is $6.92\%$ matching the Hobson-Klimmek analytical solution and the numerical solution of the dual.
Figures~\ref{fig:TransportUB0.9} and~\ref{fig:TransportUB0.7} illustrate the transport plan for the upper bound case and strikes $\Kf=0.7$ and $\Kf=0.9$.
As the strike decreases from at-the-money, more and more mass is left in place 
(starting from the left tail), and the residual mass of $\mu$ is mapped to $\nu$ via two increasing functions;
one maps the residual mass to the left tail of $\nu$ while the other maps the residual mass to the right tail of $\nu$. 
For strikes greater than at-the-money a mirror-image transport plan emerges 
where more and more mass is left in place (starting from the right tail) 
and again the residual mass of $\mu$ is mapped to $\nu$ via two increasing functions (for brevity we omit the plots).
Figures~\ref{fig:TransportLB1.05} and~\ref{fig:TransportLB1.3} illustrate the transport plan for the lower bound case and strikes $\Kf=1.05$ and $\Kf=1.3$.
As the strike increases from at-the-money,  less and less mass is left in place 
(removing mass first from the right tail) and the residual mass of $\mu$ is mapped to $\nu$ via two functions:
one maps the residual mass to the left tail of $\nu$, the other maps the residual mass to the right tail of $\nu$. 
These functions appear to be increasing for large strikes (Figure~\ref{fig:TransportLB1.3}(b)), 
but since the transport maps are decreasing for the at-the-money strike (Figure~\ref{fig:LowerBoundATM}(b)),
for strikes close to the money these maps could be decreasing~\ref{fig:TransportLB1.05}(b).
For strikes lower than the money a mirror-image transport plan emerges 
where less and less mass stays in place (removing mass first from the left tail) 
and again the residual mass of $\mu$ is mapped to $\nu$ via two functions (for brevity we omit the plots).

\section{Summary and Conclusion}
In this article, we endeavoured to provide a quantitative answer to the question as to whether 
forward-start options could be effectively replicated using Vanilla products.
The take-away message here is that they should rather be thought of as fundamental building blocks, 
and that trying to replicate them using European options is not reasonable.
Our approach using infinite linear programming arguments, 
makes the methodology directly amenable to computation and calibration to market data;
we in particular propose a discretisation scheme to reduce its (infinite) dimensionality, and therefore its complexity.
Alternatively, in line with the current active research on robust finance, 
one could rephrase the problem---using duality arguments---into 
an optimisation over sets of (martingale) measures, consistent with market data. 
Several results already exist in that direction, and we leave this for further study.


\end{document}